%% file: main.tex
\documentclass{article}

\usepackage{color}
\usepackage{authblk}
\usepackage{amsmath}
\usepackage{graphicx}
\usepackage{geometry}
\usepackage{enumerate}
\usepackage{url} 
\usepackage{subcaption}
\usepackage{multirow}
\usepackage{hyperref}
\hypersetup{
     colorlinks=true,
     linkcolor=blue,
     citecolor=magenta,      
     urlcolor=cyan,
     pdfpagemode=FullScreen,
}
\usepackage[font=footnotesize]{caption}

\usepackage{multibib}
\input{macros}

\title{\scshape Synthetic Interventions}
\author[1]{Anish Agarwal}
\author[2]{Devavrat Shah}
\author[3]{Dennis Shen\textsuperscript{*}}
\affil[1]{\small Department of Industrial Engineering \& Operations Research, Columbia University}
\affil[2]{\small Department of Electrical Engineering \& Computer Science, MIT}
\affil[3]{\small Department of Data Sciences \& Operations, USC}
\date{\today}

\begin{document}

\maketitle
\input{content/abstract}

\vspace{12pt}

\renewcommand{\thefootnote}{\fnsymbol{footnote}}
\footnotetext[1]{We sincerely thank Alberto Abadie, Abdullah Alomar, Peter Bickel, Victor Chernozhukov, Romain Cosson, Peng Ding, Esther Duflo, Avi Feller, Guido Imbens, Anna Mikusheva, Jasjeet Sekhon, Rahul Singh, and Bin Yu for invaluable feedback, guidance, and discussion. We also thank members within the MIT Economics department and the Laboratory for Information and Decisions Systems (LIDS) for useful discussions.}
\renewcommand{\thefootnote}{\arabic{footnote}}

\newpage 
\tableofcontents

\newpage
%
\input{content/intro}

\input{content/setup}

\input{content/algo}
\input{content/formal}

\input{content/simulations_tldr}

\input{content/empirical}

\input{content/conclusion}

\bibliographystyle{alpha}
\bibliography{bibliography.bib}

\begin{appendix} 

\clearpage 
\renewcommand{\theequation}{S.\arabic{equation}}
\setcounter{equation}{0}

\bigskip
\begin{center}
{\large\bf SUPPLEMENTARY MATERIAL}
\end{center}
The supplementary material is structured as follows. 
Appendices~\ref{sec:identification}, \ref{sec:helper.lemmas}, and \ref{sec:proof_consistency} prove Theorems~\ref{thm:identification} and \ref{thm:consistency}.  
Appendix~\ref{sec:nonlinear_lvm} expounds on Remark~\ref{remark:nonlinear_lvms} regarding non-linear latent variable models and low-rank factor models. 
Appendix~\ref{sec:covariates} extends the \SI~framework to incorporate auxiliary covariates. 
Appendix~\ref{sec:tensor_framework} continues our discussion on a connection between causal inference and tensor completion. 
Appendix~\ref{sec:proof.normality} elaborates on Remark~\ref{remark:normality} regarding inference and a modified version of the \SI-\PCR~estimator. 

\input{content/proofs_identification}
\input{content/proofs_interpretations}

\input{content/proofs_consistency_v2}

\input{content/nonlinear_lvm}  
\input{content/covariates}

\input{content/causal_inference_tensor_completion}

\input{content/proofs_normality}

\end{appendix}

\end{document}

%% file: macros.tex
\usepackage[utf8]{inputenc} 
\usepackage{hyperref}       
\usepackage{nicefrac}       
\usepackage{microtype}      

\usepackage{amssymb}

\usepackage{amsmath, amsfonts, amssymb, amsthm}
\usepackage{enumitem}
\usepackage{makecell}
\usepackage{color}
\usepackage{booktabs}
\usepackage{mathtools}
\mathtoolsset{showonlyrefs}

\newtheorem{theorem}{Theorem}
\newtheorem{proposition}{Proposition}
\newtheorem{lemma}{Lemma}
\newtheorem{definition}{Definition}

\newtheorem{assumption}{Assumption}

\newtheorem{remark}{Remark}

\usepackage[colorinlistoftodos]{todonotes}
\usepackage{bm}
\usepackage{bbm}
\usepackage{physics}
\usepackage{enumitem}
\usepackage{bbold}
\usepackage{epigraph}

\usepackage{csquotes}
\usepackage{url}            
\usepackage{booktabs}       
\usepackage{nicefrac}       
\usepackage{microtype}      
\usepackage{color}
\usepackage{scalerel}
\usepackage[toc,page]{appendix}
\usepackage[font=small,labelfont=bf]{caption}
\usepackage{blindtext}
\usepackage{dsfont}
\allowdisplaybreaks
\usepackage{float}
\usepackage{graphicx}
\usepackage{wrapfig}
\usepackage{subcaption}
\usepackage{subfloat}
\graphicspath{ {images/} }
\usepackage{letltxmacro}%
\usepackage{relsize}
\usepackage{algorithm,algpseudocode}
\usepackage{multirow}
\usepackage{caption} 
\usepackage[export]{adjustbox}
\usepackage{mathtools}
\mathtoolsset{showonlyrefs}
\usepackage{tikz}
\newcommand{\Var}{\mathrm{Var}}

\newcommand{\Hc}{\mathcal{H}}

\newcommand{\Pc}{\mathcal{P}}
\newcommand{\Tc}{\mathcal{T}}

\newcommand{\Ec}{\mathcal{E}}

\newcommand{\Ic}{\mathcal{I}}

\newcommand{\Ex}{\mathbb{E}}
\newcommand{\Rb}{\mathbb{R}}
\newcommand{\Nb}{\mathbb{N}}
\newcommand{\Pb}{\mathbb{P}}

\newcommand{\bU}{\boldsymbol{U}}

\newcommand{\bV}{\boldsymbol{V}}

\newcommand{\bX}{\boldsymbol{X}}
\newcommand{\bZ}{\boldsymbol{Z}}
\newcommand{\bY}{\boldsymbol{Y}}

\newcommand{\bE}{\boldsymbol{E}}

\newcommand{\bM}{\boldsymbol{M}}
\newcommand{\bA}{\boldsymbol{A}}
\newcommand{\bB}{\boldsymbol{B}}

\newcommand{\bI}{\boldsymbol{I}}
\newcommand{\bS}{\boldsymbol{S}}

\newcommand{\hbeta}{\widehat{\beta}}

\newcommand{\hw}{\widehat{w}}

\newcommand{\htheta}{\widehat{\theta}}

\newcommand{\hs}{\widehat{s}}
\newcommand{\hv}{\widehat{v}}
\newcommand{\hu}{\widehat{u}}

\newcommand{\hV}{\widehat{V}}
\newcommand{\hU}{\widehat{U}}

\newcommand{\cI}{\mathcal{I}}

\newcommand{\cE}{\mathcal{E}}

\newcommand{\bhU}{\widehat{\bU}}

\newcommand{\bhS}{\widehat{\bS}}

\newcommand{\bhV}{\widehat{\bV}}

\newcommand{\pre}{\text{pre}}
\newcommand{\epre}{\emph{pre}}
\newcommand{\post}{\text{post}}

\newcommand{\Dc}{\mathcal{D}}

\newcommand{\Wc}{\mathcal{W}}

\newcommand\independent{\protect\mathpalette{\protect\independenT}{\perp}}
\def\independenT#1#2{\mathrel{\rlap{$#1#2$}\mkern2mu{#1#2}}}

\newcommand{\espan}{\emph{span}}
\newcommand{\txtop}{\text{op}}
\newcommand{\eop}{\emph{op}}

\newcommand{\tw}{\widetilde{w}}

\newcommand{\numdonors}{N_d} 
\newcommand{\hsigma}{\widehat{\sigma}}

\newcommand{\cLF}{\mathcal{LF}} 

\newcommand{\Ynpre}{Y_{\pre, n}}

\newcommand{\Yipre}{Y_{\pre, i}}

\newcommand{\bYpred}{\bY_{\pre,  \Ic^{(d)}}}
\newcommand{\bMpred}{\bM_{\pre,  \Ic^{(d)}}}
\newcommand{\bYpostd}{\bY_{\post,  \Ic^{(d)}}}
\newcommand{\bYepred}{\bY_{\epre,  \Ic^{(d)}}}
\newcommand{\bMepred}{\bM_{\epre,  \Ic^{(d)}}}

\newcommand{\bYpredrpre}{\bY^{r_{\pre}}_{\pre, \Ic^{(d)}}}

\newcommand{\bYepredrpre}{\bY^{r_{\epre}}_{\epre, \Ic^{(d)}}}

\newcommand{\SI}{\textsf{SI}}
\newcommand{\SC}{\textsf{SC}} 
 
\newcommand{\PCR}{\texttt{PCR}}
\usepackage{pifont}
\newcommand{\xmark}{\text{\ding{55}}} 
\newcommand{\cmark}{\text{\ding{51}}}
\newcommand{\Ypostd}{Y_{\post, \cI^{(d)}}}
\newcommand{\bPhi}{\boldsymbol{\Phi}}

%% file: content/abstract.tex

The synthetic controls (\SC) methodology is a prominent tool for policy evaluation in panel data applications. 
Researchers commonly justify the \SC~framework with a low-rank matrix factor model that assumes the potential outcomes are described by low-dimensional unit and time specific latent factors. 
In the recent work of \cite{abadie_survey}, one of the pioneering authors of the \SC~method posed the question of how the \SC~framework can be extended to multiple treatments. 
This article offers one resolution to this open question that we call synthetic interventions (\SI). 
Fundamental to the \SI~framework is a low-rank tensor factor model, which extends the matrix factor model by including a latent factorization over treatments. 
Under this model, we propose a generalization of the standard \SC-based estimators. 
We prove the consistency for one instantiation of our approach and  provide conditions under which it is asymptotically normal. 
Moreover, we conduct a representative simulation to study its prediction performance and  revisit the canonical \SC~case study of \cite{abadie2} on the impact of anti-tobacco legislations by exploring related questions not previously investigated.

%% file: content/intro.tex
\section{Introduction}\label{sec:intro}
In 1988, California passed the first large-scale anti-tobacco program in the United States known as Proposition 99.
Early reports of the success of Proposition 99 sparked a wave of anti-tobacco measures across the country: from 1989--2000, four states (Arizona, Massachusetts, Oregon, and Florida) introduced similar anti-tobacco  programs and seven states (Alaska, Hawaii, Maryland, Michigan, New Jersey, New York, Washington) raised their state cigarette taxes by 50 cents or more. 
The remaining 38 states (e.g., Kansas or Virginia) kept their statewide status quos. 
To isolate the impact of these anti-tobacco measures, researchers sought an answer to the following question: 
%
\begin{center}
{\bf Q1}: {\em what would have happened if a ``treated'' state that either imposed an anti-tobacco program or raised taxes had kept their statewide status quo?} 
\end{center} 

The seminal papers of \cite{abadie1, abadie2} introduced the synthetic controls (\SC) framework that provides an elegant approach to tackle this challenge.
In the case of California, the \SC~method constructs a synthetic California from a weighted composition of ``control'' states that did not impose any anti-tobacco measures to estimate California's counterfactual cigarette sales (outcome variable of interest) in the absence of Proposition 99.
The weights are chosen such that the trajectory of cigarette sales of the synthetic California mirrors that of the observed California prior to the passage of Proposition 99 (``pre-treatment'' period). 
The difference between the observed California with Proposition 99 and the synthetic California without Proposition 99 during the ``post-treatment'' period of 1989--2000 sheds insight into the causal effect of the initiative.
Following this procedure, \cite{abadie2} found that tobacco consumption fell markedly in California under Proposition 99 compared to its synthetic counterpart.  
Given its immense popularity, the \SC~framework is often regarded as one of the most important innovations in the policy evaluation literature \cite{athey}.

However, even though the \SC~framework offers a resolution to the question above, several related questions remain: {\em what would have happened if a...}
\begin{center}
	{\bf Q2}: {\em treated state that imposed an anti-tobacco program had raised taxes?}
	\\ 
	{\bf Q3}: {\em treated state that raised taxes had imposed an anti-tobacco program?}
	\\ 
	{\bf Q4}: {\em control state that kept their status quo had raised taxes?}
	\\ 
	{\bf Q5}: {\em control state that kept their status quo had imposed an anti-tobacco program?}
\end{center} 
The \SC~framework, in its current form, is not equipped to explore these set of questions. 
To see this, consider how \SC~would address Q2. 
Following Q1, the \SC~method constructs a synthetic California from a weighted composition of {\em treated} states that raised taxes during the post-treatment period, states such as Kansas and Virginia. 
These weights, as in Q1, would be learned during the pre-treatment period whence all outcomes are observed under {\em control}.
In contrast to Q1, however, these weights would then be applied on the {\em treated} outcomes associated with those states that raised taxes during the post-treatment period. 
Similar approaches can be taken to answer Q3-Q5, e.g., the solution to Q4 would mimic that of Q2 with the control states in place of the anti-tobacco program states. 

The juxtaposition of Q1 and Q2-Q5 reveals their fundamental difference:
whereas the \SC~solution to Q1 applies the weights learnt from the pre-treatment outcomes observed under {\em control} to post-treatment outcomes also observed under {\em control}, the \SC~solutions to Q2-Q5 apply the weights learnt from the pre-treatment outcomes observed under {\em control} to post-treatment outcomes observed under {\em treatment}. 
Restated differently, Q1 inherits ``data symmetry'' between the pre- and post-treatment periods while Q2-Q5 faces ``data asymmetry'' between the pre- and post-treatment periods. 
The stark contrast in data usage between Q1 versus Q2-Q5 begs the question: can the \SC~framework be extended to the setting of multiple treatments such that the same data used to answer questions a la Q1 can also be used to answer questions a la Q2-Q5? 
More formally, can a set of weights learned from pre-treatment outcomes observed under one setting (e.g., control) be applied toward post-treatment outcomes observed under another setting (e.g., treatment)? 
This is posed as an open question by one of the pioneering authors of the \SC~framework \cite{abadie_survey}. 


\subsection{Contributions}
In this article, we propose the synthetic interventions (\SI) framework, which provides one resolution to the question asked by \cite{abadie_survey}. 
Building on the canonical matrix factor model that is widely used to justify the \SC~framework, we consider a {\em tensor} factor model, which serves as the bedrock for the \SI~framework. 

Intuitively, the matrix factor model postulates that the potential cigarette sales of California under their statewide status quo (i.e., absence of any anti-tobacco measure) in the year 1993, for instance, can be characterized by latent factors that describe California and the year 1993. 
Notably, only unit (e.g., California) and time (e.g., year 1993) specific latent factors are required since the \SC~framework is solely concerned with outcomes observed under a single setting, typically control. 
The tensor factor model, on the other hand, postulates that the same potential outcome is characterized by not only the latent factors of California and the year 1993, but also the latent factor associated with the status quo (control). 
More generally, the tensor factor model posits that each potential outcome is described by three latent factors tailored to the unit, time period, and treatment of interest, where control can be viewed as the ``default'' treatment. 
With the added assumption of separate latent factors for treatments, the \SI~framework can justifiably exploit the same data pattern considered in the \SC~framework yet provide answers to far more questions.  

Under the tensor factor model, we establish the identification of several causal estimands related to questions such as Q2-Q5. 
Moreover, we propose a generalization of the standard \SC-based estimators to answer said questions of interest. 
We analyze one instantiation of this methodology and prove its consistency. 
In Appendix~\ref{sec:proof.normality}, we provide conditions under which our estimation strategy is asymptotically normal and show that a minor modification to the instantiation analyzed in the main body satisfies those conditions and thus enables inference. 
Notably, our results remain valid in the presence of hidden confounders, provided these confounders can be explained by the latent factors associated with our tensor formulation. 
We corroborate our theoretical results with simulations and revisit the popular California Proposition 99 study of \cite{abadie2} by examining the other treated states; 
in the latter, we find that (i) anti-tobacco programs and raised taxes had similar effects and (ii) both measures led to a consequential decrease in tobacco consumption relative to the status quo, which reaffirms the conclusion drawn in \cite{abadie2} for California.

\subsection{Organization of Paper}\label{sec:organization} 
Section~\ref{sec:setup} presents the \SI~framework and describes the causal estimand. 
Section~\ref{sec:algo} proposes an extension of the \SC-based estimators for our setting. 
Section~\ref{sec:formal} analyzes one instantiation of our estimator and proves its consistency. 
Section~\ref{sec:simulations} reinforces our theoretical results with a simulation study. 
Section~\ref{sec:empirics} revisits the classical California Proposition 99 study of \cite{abadie2}. 
Section~\ref{sec:conclusion} summarizes this article. 
Additional discussions and proofs are relegated to the appendix.

\subsection{Notation}\label{sec:notation}
For any positive integer $a$, let $[a] = \{1, \dots, a\}$ and $[a]_0 = \{0, \dots, a-1\}$. 
For a vector $v \in \Rb^a$, let $\|v\|_p$ denote its 
$\ell_p$-norm. 
We define the inner product between vectors $v, x \in \Rb^a$ as 
$\langle v, x \rangle = v^\top x = \sum_{\ell=1}^a v_\ell x_\ell$. 
For a matrix $\bA \in \Rb^{a \times b}$, we denote its Frobenius norm as $\|\bA\|_F$. 
Let $\| \cdot \|_{\psi_2}$ denote the Orlicz norm. 
Let $O_p$ and $o_p$ denote the probabilistic versions of the deterministic big-$O$ and little-$o$ notations. 
Denote $\bI$ as the identity matrix and $(\cdot)^{\dagger}$ as the Moore-Penrose pseudoinverse.

%% file: content/setup.tex
\section{Synthetic Interventions Framework}\label{sec:setup}
We consider the canonical data format within the \SC~literature known as {\em panel data}, a collection of observations on multiple units over multiple time periods. 
Formally, we observe $N$ units over $T$ time periods, indexed by $i \in [N]$ and $t \in [T]$, respectively. 
As an example, the tobacco case study tracks the cigarette sales of $N = 50$ states over $T = 31$ years. 
Following the potential outcomes framework of \cite{neyman, rubin}, we philosophize that in each time period $t$, each unit $i$ is characterized by $D$ potential outcomes denoted as $Y_{ti}^{(d)} \in \Rb$ for $d \in [D]_0$; 
consistent with standard notation, we reserve $d = 0$ for control. 
In the context of the tobacco study, the potential outcomes framework asserts that three versions of each state simultaneously exist each year: one that keeps the statewide status quo ($d=0$), one that imposes an anti-tobacco program ($d=1$), and one that raises taxes ($d=2$). 
Of course, each state can only adopt one policy in any year and hence, researchers can never observe all possible potential outcomes---this is the fundamental challenge of causal inference.

\subsection{Counterfactual Estimation as Tensor Completion} \label{sec:ci.tc}
%
The synthetic interventions (\SI) framework views the problem of counterfactual estimation through the lens of an order-three tensor. 
%
Specifically, we encode the potential outcomes into a $T \times N \times D$ tensor, $\bY^* = [Y_{ti}^{(d)}]$, whose dimensions correspond to time, units, and treatments. 
%
Accordingly, the $(t, i, d)$-th entry of $\bY^*$ is the potential outcome of unit $i$ at time $t$ under treatment $d$. 
If $\bY^*$ is known to the researcher, then any causal quantity of interest can be calculated by accessing the relevant entries of $\bY^*$. 
The challenge, however, is that the researcher can only observe a sparse subset of $\bY^*$. 

In the classical \SC~setup, all $N$ units are observed under control for the first $T_0$ time periods. 
We refer to this duration as the pre-intervention or pre-treatment period. 
For the remaining $T_1 = T - T_0$ time periods, known as the post-intervention or post-treatment period, each unit is assigned to one of the $D$ treatments. 
Let $\Ic^{(d)} \subset [N]$ denote the set of units that are assigned to treatment $d$ during the post-treatment period with $N_d = |\Ic^{(d)}|$. 
In the tobacco study, $N_0 = 38$ states kept their statewide status quo, $N_1=5$ states imposed a tobacco control program, and $N_2=7$ states raised taxes during the post-treatment period. 

We denote $Y_{tid} \in \{\Rb \cup \star\}$, where $\star$ represents a missing entry, as the recorded outcome for unit $i$ at time $t$ under treatment $d$, and store these outcomes into a $T \times N \times D$ tensor $\bY = [Y_{tid}]$. 
The following assumption formalizes our observation pattern in $\bY$ and establishes its connection with $\bY^*$. 

\begin{assumption}[Stable Unit Treatment Value Assumption]\label{assumption:sutva}
We observe 
\begin{itemize}
	\item[(i)] Pre-treatment period: $Y_{ti0} = Y_{ti}^{(0)}$ for $t \le T_0$ and $i \in [N]$. 
	
	\item[(ii)] Post-treatment period: $Y_{tid} = Y_{ti}^{(d)}$ for $t > T_0$, $i \in \Ic^{(d)}$, and $d \in [D]_0$. 
\end{itemize} 
For all other entries of $\bY$, we have $Y_{tnd} = \star$.
\end{assumption} 
In view of Assumption~\ref{assumption:sutva}, estimating counterfactual potential outcomes is equivalent to imputing missing entries in $\bY$; in the machine learning community, the latter problem is known as {\em tensor completion}. 
For a visualization of $\bY^*$ and $\bY$, see Figure~\ref{fig:tensor}. 
We remark that Assumption~\ref{assumption:sutva} implicitly rules out spillover (network) effects between units \cite{imbens_rubin_2015}. 

\begin{figure}[t]
	\centering 
	\begin{subfigure}[b]{0.33\textwidth}
		\centering 
		\includegraphics[width=\linewidth]{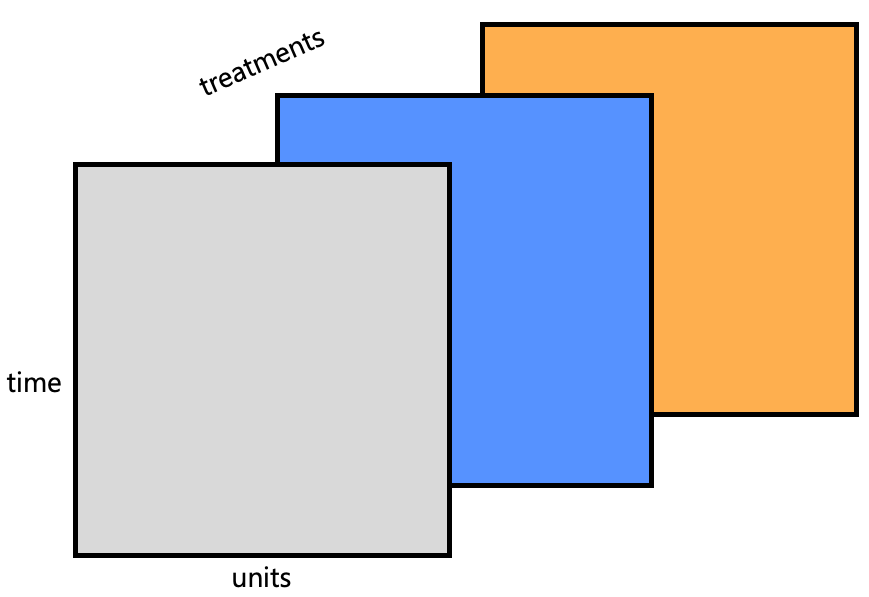}
		\caption{Potential outcomes tensor, $\bY^*$.} 
		\label{fig:tensor.ideal} 
	\end{subfigure} 
	\qquad \qquad
	\begin{subfigure}[b]{0.33\textwidth}
		\centering 
		\includegraphics[width=\linewidth]{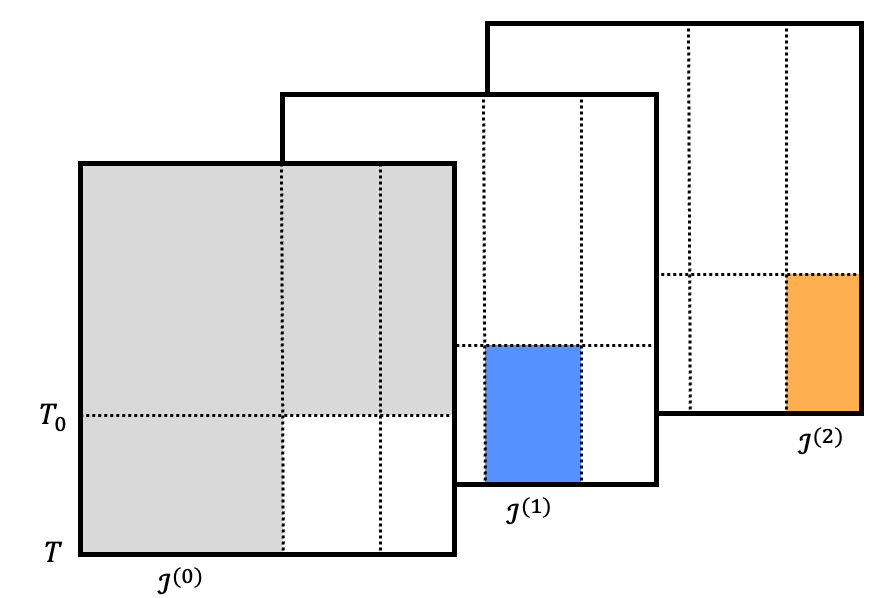}
		\caption{Observed tensor, $\bY$.} 
		\label{fig:tensor.obs} 
	\end{subfigure} 
	\caption{Figure~\ref{fig:tensor.ideal} visualizes the potential outcomes tensor, $\bY^*$, while Figure~\ref{fig:tensor.obs} visualizes the observed tensor, $\bY$, as per Assumption~\ref{assumption:sutva} and tailored to our tobacco study. 
	The colored blocks indicate observed entries with the color indexing the treatment: status quo in gray, anti-tobacco programs in blue, and raised taxes in orange. 
	The white blocks indicate missing entries.}
	\label{fig:tensor} 
\end{figure}

\subsection{Tensor Factor Model}\label{ssec:tfm}
To enable faithful recovery of the missing entries in $\bY$, we impose structure on $\bY^*$. 
We motivate our model by first revisiting the dominant idea behind the \SC~framework---the low-rank matrix factor model, also known as the interactive fixed effects model \cite{bai03}. 
Recall that the \SC~framework is solely focused on a single treatment setting, most commonly control. 
Restated through our tensor formulation, the \SC~framework aims to recover the single tensor slice (matrix) corresponding to control. 
The potential outcomes under control are posited to obey the following: for all $t \in [T]$ and $i \in [N]$, we have 
\begin{align} \label{eq:matrix.fm} 
	Y_{ti}^{(0)} = \langle u_t, v_i \rangle + \varepsilon_{ti} = \sum_{\ell = 1}^r u_{t\ell} v_{i\ell} + \varepsilon_{ti}, 
\end{align} 
where (i) $u_t \in \Rb^r$ and $v_t \in \Rb^r$ are latent factors specific to time $t$ and unit $i$, respectively, with $r$ sufficiently small relative to $N_0$ and $T_0$, and (ii) $\varepsilon_{ti} \in \Rb$ represents the (mean-zero) idiosyncratic shock. 
The key structure in \eqref{eq:matrix.fm} is that the latent unit factors are {\em invariant} across time. 
As such, the relationship between units in the pre-treatment period continues to hold in the post-treatment period. 
In turn, researchers can learn a set of weights from pre-treatment outcomes observed under one setting (e.g., status quo) and apply them toward post-treatment outcomes observed under the {\em same} setting (e.g., status quo) in order to answer questions a la Q1 of Section~\ref{sec:intro}. 
The \SC~framework is thus justified under \eqref{eq:matrix.fm} \cite{abadie_survey}. 

The tensor factor model builds upon \eqref{eq:matrix.fm} by introducing a latent factorization across treatments: for all $t \in [T]$, $i \in [N]$, and $d \in [D]_0$, we have 
\begin{align} \label{eq:tensor.fm}
	Y_{ti}^{(d)} = \sum_{\ell = 1}^r u_{t\ell} v_{i\ell} \lambda_{d\ell} + \varepsilon_{ti}^{(d)}, 
\end{align} 
where (i) $\lambda_d \in \Rb^r$ is a latent factor specific to treatment $d$ with $r$ sufficiently small relative to $N_d$ and $T_0$, and (ii) $\varepsilon_{ti}^{(d)}$ is the (mean-zero) idiosyncratic shock. 
Under the tensor factor model, the latent unit factors are invariant across both time {\em and treatments}. 
As a result, the relationship between units in the pre-treatment period under one treatment continues to hold in the post-treatment period under {\em another treatment}. 
Translated to practice, researchers can learn a set of weights from pre-treatment outcomes observed under one setting (e.g., status quo) and apply them toward post-treatment outcomes observed under {\em another setting} (e.g., raised taxes) to answer questions a la Q2-Q5 of Section~\ref{sec:intro}. 
The \SI~framework, therefore, is justified under \eqref{eq:tensor.fm}. 
We highlight that the low-rank tensor factor model stands as the de-facto assumption in the tensor completion literature \cite{tensor_missing, gandy, anandkumar2014tensor, barak2016noisy}. 

With this concept in mind, we present our primary structural assumption on $\bY^*$. 

\begin{assumption}[Tensor factor model] \label{assumption:form} 
Each entry in $\bY^*$ satisfies the following: for all $(t,i,d)$, we have 
\begin{align} \label{eq:form.0} 
Y_{ti}^{(d)} = \langle u^{(d)}_t,  v_i \rangle + \varepsilon_{ti}^{(d)} = \sum_{\ell=1}^r u^{(d)}_{t\ell} v_{i\ell} + \varepsilon_{ti}^{(d)}, 
\end{align} 
where $u^{(d)}_t \in \Rb^r$ is a latent factor specific to $(t, d)$, $v_i \in \Rb^r$ is a latent factor specific to $i$, and $\varepsilon_{ti}^{(d)} \in \Rb$ is the (mean-zero) idiosyncratic shock. 
\end{assumption} 

Note that \eqref{eq:form.0} is implied by \eqref{eq:tensor.fm} since we can express $u_{t\ell}^{(d)} = u_{t\ell} \lambda_{d\ell}$, where $u_{t\ell}$ and $\lambda_{d\ell}$ are defined as in \eqref{eq:tensor.fm}. 
Since the invariance of the unit factors is the critical characteristic needed to validate the \SI~framework, we operate with the weaker condition stated in \eqref{eq:form.0} as opposed to \eqref{eq:tensor.fm}. 
We acknowledge that \eqref{eq:form.0} adds to the traditional matrix factor model assumption within the \SC~framework by imposing that each slice of $\bY^*$ satisfies its own matrix factor model and that the slices share an invariant set of latent unit factors. 

\subsection{Causal Estimand} 
Anchoring on the tensor factor model stated in Assumption~\ref{assumption:form}, we are ready to define our causal estimand.   
Formally, for a given unit $i$ and treatment $d$, we are interested in estimating  
\begin{align} \label{eq:causal.est}
	\theta_i^{(d)} = \frac{1}{T_1} \sum_{t > T_0} \Ex \left[ Y_{ti}^{(d)} | \{u_t^{(d)}, v_i: t > T_0\} \right].
\end{align} 
In words, \eqref{eq:causal.est} represents the average expected potential outcome for unit $i$ under treatment $d$ during the post-treatment period. 
Since \eqref{eq:causal.est} conditions on the latent factors, its expectation is taken over the measurement shocks. 
Returning to Q2 in Section~\ref{sec:intro} with California as the unit of interest, \eqref{eq:causal.est} translates as California's average expected cigarette sales from 1989--2000 had the state raised taxes instead of imposing an anti-tobacco program. 

We take this opportunity to underscore once more the contrast between the \SC~and \SI~objectives: 
whereas the \SC~framework aims to recover $\theta_i^{(0)}$, the \SI~framework aims to recover $\theta_i^{(d)}$ for any treatment $d$. 
Put another way, the \SC~objective is to impute the first tensor slice in Figure~\ref{fig:tensor.obs} while the \SI~framework looks to impute the remaining slices during the post-treatment period.  
Next, we discuss how our causal estimand can be identified. 

\subsubsection{Additional Assumptions for Identification}
Let $\Dc = \{ (t, i, d) : ~Y_{tid} \neq \star \}$ denote the treatment assignments. 
Ideally, $\Dc$ is independent of $\bY^*$ as in the case of randomized experiments, the gold standard mechanism for learning causal effects. 
With observational data, however, $\Dc$ is often dependent on $\bY^*$---this is known as {\em confounding}. 
For example, a state policy-maker may favor an anti-tobacco program as they expect their statewide tobacco consumption to fall more rapidly under the program than under the continued status quo or raised taxes. 

We consider a setting where the confounders can be latent but their impact on $\Dc$ and $\bY^*$ is mediated through the latent factors, $\cLF = \{ u^{(d)}_t, v_i: t \in [T], i \in [N], d \in [D]_0 \}$. 
%

\begin{assumption} [Selection on latent factors] \label{assumption:conditional_mean_zero}
For all $(t,i,d)$, we have $\Ex[\varepsilon^{(d)}_{ti}| \cLF] = 0$ and $\varepsilon^{(d)}_{ti} \independent \Dc ~|~ \cLF$. 
\end{assumption} 
Strictly speaking, we only require $\Ex[\varepsilon^{(d)}_{tn}| \cLF, \Dc]  = 0$, which is known as conditional mean independence. 
We state it as $\varepsilon^{(d)}_{tn} \independent \Dc ~|~ \cLF$ to increase the interpretability of our conditional exogeneity assumption.
Collectively, Assumptions \ref{assumption:form} and \ref{assumption:conditional_mean_zero} imply that 
$\bY^* \independent \Dc ~|~ \cLF$. 
Hence, Assumption~\ref{assumption:conditional_mean_zero} can be interpreted as {\em selection on latent factors}, which is analogous to the classical {\em selection on observables} assumption. 
That is, Assumption~\ref{assumption:conditional_mean_zero} states that the latent factors can act as our latent confounders. 
Similar conditional independence assumptions have been utilized in \cite{athey1}, \cite{kallus2018causal}, and \cite{asc}.

To state our next assumption, let  $\cE = \{\cLF, \Dc\}$. 

\begin{assumption} [Latent unit factors] \label{assumption:linear}
Given $(i,d)$, we have $v_i \in \espan(\{v_j: j \in \Ic^{(d)}\})$, conditional on $\cE$. 
Specifically, there exists weights $w^{(i,d)} \in \Rb^{\numdonors}$ such that $v_i = \sum_{j \in \Ic^{(d)}} w^{(i,d)}_j  v_j$. 
\end{assumption} 
Assumption~\ref{assumption:linear} states that, within the space of latent unit factors, unit $i$ can be reconstructed as a weighted average of units that adopt treatment $d$ during the post-treatment period. 
The critical feature to observe is that $w^{(i,d)}$ is time- and treatment-invariant. 
%
%
Returning to our tobacco study, if Assumption~\ref{assumption:linear} holds for California under raised taxes ($d=2$), then California's latent factor can be expressed as a weighted sum of latent factors for the seven states that raised taxes from 1989--2000. 
If Assumption~\ref{assumption:linear} were to also hold between California under the status quo ($d=0$), then California's latent factor can equivalently be written as another weighted sum of latent factors associated with the 38 states that kept their statewide status quo from 1989--2000. 
Due to its invariance property, the weights in both settings remain static regardless of the time period or treatment. 

When the tensor factor model stated in Assumption~\ref{assumption:form} admits a low-dimensional structure, i.e., $r$ as defined in \eqref{eq:form.0} is sufficiently small (formalized in Theorem~\ref{thm:consistency}), Assumption~\ref{assumption:linear} effectively follows as an immediate consequence. 
To see this, let $\bV_{\Ic^{(d)}} = [v_j: j \in \Ic^{(d)}] \in \Rb^{N_d \times r}$. 
Note that if $\text{rank}(\bV_{\Ic^{(d)}}) = r \ll N_d$, then Assumption~\ref{assumption:linear} becomes a direct byproduct.  
In other words, a low-dimensional tensor factor model implies that the latent unit factors are linearly dependent, which is the structure exploited in Assumption~\ref{assumption:linear}. 
From a practical standpoint, Assumption~\ref{assumption:linear} requires that sufficiently many units are assigned to treatment $d$ and that their latent factors are collectively ``diverse'' enough so that their span includes the latent factor of unit $i$. 
In this view, Assumption~\ref{assumption:linear} is analogous to the classical {\em common support} assumption. 
We finish by remarking that low-dimensional factor models have been shown to be a ubiquitous phenomenon in practice and are justified under many natural data generating processes \cite{Chatterjee15, Udell2017NiceLV, udell2018big, xu2017rates}. 


\subsubsection{Identification Result} 
With our assumptions stated, we establish our identification result for \eqref{eq:causal.est}. 

\begin{theorem}\label{thm:identification}
Given $(i,d)$, let Assumptions \ref{assumption:sutva} to \ref{assumption:linear} hold.
Then we have 
\begin{align}
\Ex [ Y^{(d)}_{ti} \big| u^{(d)}_t, v_i ] &= \sum_{j \in \Ic^{(d)}} w_j^{(i,d)} \cdot \Ex\left[ Y_{tjd} | \cE \right] \quad \text{for all}~ t \in [T], \label{eq:identification_strong}
\\ \theta^{(d)}_i &= \frac{1}{T_1} \sum_{t > T_0} \sum_{j \in \Ic^{(d)}} w_j^{(i,d)} \cdot \Ex\left[ Y_{tjd} | \cE \right]. \label{eq:identification}
\end{align}
\end{theorem}
Theorem~\ref{thm:identification} involves two claims: \eqref{eq:identification} states that our causal estimand in \eqref{eq:causal.est}, which is defined as an average over all $t > T_0$, can be written as a function of estimable quantities;  \eqref{eq:identification_strong} is a stronger version of \eqref{eq:identification} that holds for each time $t \in [T]$. 
As a consequence of Theorem~\ref{thm:identification}, for a given $(i, d)$ pair, we look to estimate  $w^{(i,d)}$, which we reemphasize is invariant across time and treatments. 
Section~\ref{sec:algo} proposes an estimator to achieve this objective. 

\subsection{Remarks on Assumptions} 
%
\begin{remark} [On Assumption~\ref{assumption:form}] \label{remark:nonlinear_lvms}
{\em 
In Appendix~\ref{sec:nonlinear_lvm}, we show that sufficiently continuous non-linear latent variable models are arbitrarily well-approximated by linear factor models of the form \eqref{eq:form.0} in Assumption \ref{assumption:form} as $N$ and $T$ grow. 
}
\end{remark} 

\begin{remark} [On Assumption~\ref{assumption:conditional_mean_zero}]
{\em 
As with any assumption on confounding, it is difficult---if not impossible---to verify Assumption~\ref{assumption:conditional_mean_zero} for observational data. 
Rationalizing the validity of Assumption~\ref{assumption:conditional_mean_zero}, therefore, requires subject matter knowledge on the application of interest. 
}
\end{remark}

\begin{remark} [On Assumption~\ref{assumption:linear}]
{\em 
Assumption~\ref{assumption:linear} cannot be fully verified since the unit factors are latent. 
With that said, we can examine its validity through the pre-treatment fit, as prescribed in \cite{abadie1, abadie2, abadie_survey}. 
See \cite{asc, SDID} for a bias-correction variant of the standard \SC-based estimators when the pre-treatment fit is poor.  
}
\end{remark}

%
%
%
%

%% file: content/algo.tex
\bigskip 
\section{Synthetic Interventions Estimator}\label{sec:algo}
Motivated by Theorem~\ref{thm:identification}, we propose the synthetic interventions (\SI) estimator. 
%
Without loss of generality, we focus on estimating $\theta^{(d)}_i$ for a given $(i,d)$ pair. 
%


\subsection{Description of Estimator}\label{ssec:method}
To facilitate the description of the \SI~estimator, we introduce some useful notation. 
Let $\Yipre = [Y_{ti0} : t \le T_0] \in \Rb^{T_0}$ collect the pre-treatment outcomes for unit $i$.
Let $\bYpred = [ Y_{tj0}  : t \le T_0, \ j \in \Ic^{(d)} ] \in \Rb^{T_0 \times \numdonors}$ and $\bYpostd = [ Y_{tjd} : t > T_0, \ j \in \Ic^{(d)} ] \in \Rb^{T_1 \times  \numdonors}$ 
collect the pre- and post-treatment outcomes, respectively, for the units in $\Ic^{(d)}$.
Note that $\bYpred$ is observed under the ``default'' treatment (control) while $\bYpostd$ is observed under the treatment $d$ of interest.
The \SI~estimator is implemented via the following two steps: 
\begin{enumerate}
    \item[(i)] {\em Model identification}: given a constraint set $\Wc$, define 
    \begin{align} \label{eq:si.linear_model}
    	\hw^{(i,d)} &\in \underset{w \in \Wc}{\text{argmin}} ~ \| \Yipre - \bYpred w \|_2^2. 
    \end{align}
    
    \item[(ii)] {\em Point estimation}: let $\widehat{\Ex}[Y^{(d)}_{ti}] = \sum_{j \in \Ic^{(d)}} \hw^{(i,d)}_j \cdot Y_{tjd}$ for each $t > T_0$. Then, we have 
    \begin{align}  
        \htheta^{(d)}_i &= \frac{1}{T_1} \sum_{t > T_0} \widehat{\Ex}[Y^{(d)}_{ti}]. \label{eq:si.2} 
    \end{align}
    
\end{enumerate}

For an illustration of the \SI~estimator, consider Q2 of Section~\ref{sec:intro} with California and raised taxes as unit $i$ and treatment $d$, respectively. 
In \eqref{eq:si.linear_model}, the \SI~estimator outputs a set of weights that minimizes the pre-treatment fit between California and a weighted average of the seven states that raised taxes (e.g., New Jersey) subject to the constraint set $\Wc$. 
The weighted combination of these taxed states acts as our synthetic California. 
Accordingly, in \eqref{eq:si.2}, the \SI~estimator proceeds to re-weight the post-treatment outcomes of those seven states as per the learnt weights to predict California's counterfactual trajectory of cigarette sales from 1989--2000 had its policy-makers opted to raise taxes. 

We reemphasize that the weights are constructed from pre-treatment outcomes observed under control and applied toward post-treatment outcomes observed under treatment $d$, which can be the same or different from control. 
In this view, the \SI~estimator is a generalization of the standard \SC-based estimators that operate only on control outcomes. 

\subsection{Different Formulations of the \SI~Estimator}
There are numerous formulations of the \SC~estimator, each characterized by a particular choice of $\Wc$ in \eqref{eq:si.linear_model}. 
%
The pioneering works of \cite{abadie1, abadie2} restrict the weights to lie within the simplex, i.e., $\Wc = \{w \in \Rb^{\numdonors}: w_j \ge 0, \sum_{j \in \Ic^{(d)}} w_j = 1\}$. 
Attractive aspects of the simplex formulation include interpretability, sparsity, and transparency \cite{abadie_survey}. 
Other common approaches include the lasso \cite{LiBell17, arco, chernozhukov2020practical, sc_lasso1}, ridge regression \cite{asc}, elastic-net \cite{imbens16}, principal component regression \cite{rsc, mrsc}, and unconstrained linear regression \cite{hcw, li2020}. 
Given its connection with the \SC~estimator, the \SI~estimator also allows for the same flexibility in defining $\Wc$. 

\subsection{A Variation Based on Principal Component Regression} \label{ssec:pcr}
In the remainder of this article, we consider the principal component regression (\PCR) variant of the \SI~estimator, which we refer as \SI-\PCR. 
To describe this variation, we denote the singular value decomposition of $\bYpred$ as $\bYpred =  \sum_{\ell \ge 1} \hs_{\ell} \hu_{\ell} \hv^\top_{\ell}$, 
where $\hs_\ell \in \Rb$ are the singular values arranged in decreasing order, and $\hu_\ell \in \Rb^{T_0}, \hv_\ell \in \Rb^{\numdonors}$ are the corresponding left and right singular vectors, respectively. 
Letting $\bhV_\pre = [\hv_1, \dots, \hv_k] \in \Rb^{\numdonors \times k}$ collect the top $k$ right singular vectors of $\bYpred$, we define $\Wc = \{w \in \Rb^{\numdonors}: (\bI - \bhV_\pre \bhV_\pre^\top) w = 0\}$.  
This yields a closed-form expression for the linear model: 
\begin{align} \label{eq:pcr}
	\hw^{(i,d)} &= \left( \sum_{\ell=1}^k (1/\hs_\ell) \hv_\ell \hu^\top_\ell \right) \Yipre. 
\end{align} 
In words, \SI-\PCR~first pre-processes $\bYpred$ by obtaining its rank-$k$ approximation and then regresses $\Yipre$ on the resulting matrix. 
%

\paragraph{Benefits.}
\PCR~yields several benefits.
To begin, \PCR~tackles overfitting by regularizing the weights through a spectral sparsity constraint on $\bYpred$ that sets all of its singular values smaller than $\hs_k$ to zero. 
Further, Assumptions~\ref{assumption:form} and \ref{assumption:conditional_mean_zero} imply that we would ideally regress $\Yipre$ on $\Ex[\bYpred | \Ec]$.
%
%
In reality, however, we only observe its noisy instantiation, $\bYpred$; this is known as {\em error-in-variables} regression. 
When $\Ex[\bYpred | \Ec]$ has low-dimensional structure, \cite{agarwal2020robustness} and \cite{pcr_aos} show that the first step of \PCR~is an effective approach to de-noise $\bYpred$. 
For this reason, \PCR~is particularly suited under our model. 

\paragraph{Choosing $k$.} 
Popular data-driven approaches to choose $k$ include cross-validation and ``universal'' thresholding schemes that preserve singular values above a precomputed value \cite{Gavish_2014, usvt}. 
A human-in-the-loop approach is to choose $k$ as the ``elbow'' point (see Figure~\ref{fig:elbow}) that partitions the singular values of $\bYpred$ into those of large and small magnitudes. 
%
%
Note that if $\bE_{\pre, \Ic^{(d)}} = \bYpred - \Ex[\bYpred | \Ec]$ has independent sub-Gaussian rows, then the singular values of $\bE_{\pre, \Ic^{(d)}}$ scale as $O_p(\sqrt{T_0} + \sqrt{N_d})$ \cite{vershynin2018high}. 
If the entries of $\Ex[\bYpred | \Ec]$ are $\Theta(1)$ and its nonzero singular values are of the same magnitude, then they will scale as  $\Theta(\sqrt{T_0 N_d / r_\pre} )$, where $r_\pre = \rank(\Ex[\bYpred | \Ec])$. 
Such a model leads to the aforementioned elbow point.
Thus, one feasibility check for the suitability of \PCR~is to empirically examine if the smallest retained singular value $\hs_k$ scales quicker than $\sqrt{T_0} + \sqrt{N_d}$. 



%
\begin{figure}[]
 \centering
 \includegraphics[width=0.45\linewidth]
 {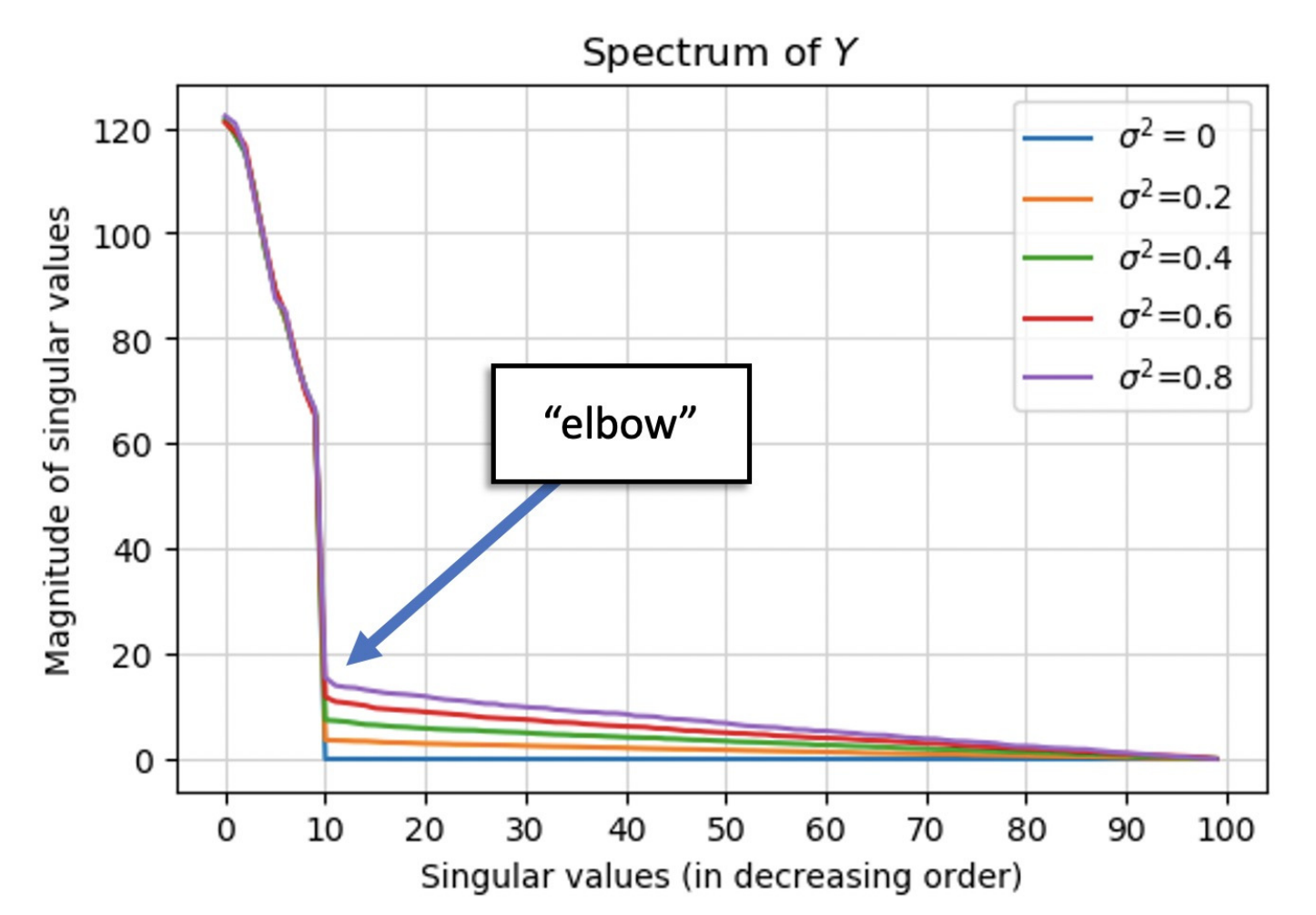}
 \caption{
 Simulation displays the spectrum of $\bY = \Ex[\bY] + \bE \in \Rb^{100 \times 100}$. 
 Here, $\Ex[\bY] = \bU \bV^\top$, where the entries of $\bU, \bV \in \Rb^{100 \times 10}$ are sampled independently from $\mathcal{N}(0,1)$;  
 the entries of $\bE$ are sampled independently from $\mathcal{N}(0, \sigma^2)$ with $\sigma^2 \in \{0, 0.2, \dots, 0.8\}$. 
 Across varying levels of $\sigma^2$, there is a steep drop-off in magnitude of the singular values---this marks the ``elbow'' point. 
 The top singular values of $\bY$ correspond closely with that of $\Ex[\bY]$ ($\sigma^2=0$), and the remaining singular values are induced by $\bE$.
 Thus, $\rank(\bY) \approx \rank(\Ex[\bY]) = 10$.}
 \label{fig:elbow}
\end{figure}

%% file: content/formal.tex
\section{Formal Results}\label{sec:formal}
In this section, we analyze the \SI-\PCR~estimator described in Section~\ref{ssec:pcr}. 

\subsection{Additional Assumptions for Consistency} 
We state additional assumptions required to establish the consistency of the \SI-\PCR~estimator. 
Please refer to Sections~\ref{sec:setup} and \ref{sec:algo} for a refresher of our earlier assumptions and notation. 

\begin{assumption} [Sub-Gaussian noise] \label{assumption:noise} 
For all $(t,i,d)$,
$\varepsilon^{(d)}_{ti}$ are independent sub-Gaussian random variables 
with $\Var(\varepsilon^{(d)}_{ti} | \cLF ) = \sigma^2$ and $\|\varepsilon^{(d)}_{ti} | \cLF \|_{\psi_2} \le C\sigma$ for some constant $C > 0$.
\end{assumption}

\begin{assumption} [Boundededness] \label{assumption:boundedness}
For all $(t,i,d)$, we have $\Ex[Y_{ti}^{(d)} | \cLF ] \in [-1,1]$.
\end{assumption}

\begin{assumption} [Well-balanced spectra] \label{assumption:spectra} 
Given $d$, we have $\kappa^{-1} \ge c$  and $\| \Ex[\bYepred | \cE] \|_F^2 \ge c' \numdonors T_0$, where $\kappa$ is the condition number of  $\Ex[\bYepred | \cE]$ and and constants $c, c' > 0$. 
\end{assumption}

\begin{assumption} [Latent time-treatment factors] \label{assumption:subspace} 
Given $d$, we have $u_t^{(d)} \in \espan(\{u^{(0)}_\tau: \tau \le T_0\})$ for every $t > T_0$, conditional on $\cE$. 
\end{assumption} 
We defer our discussions of Assumptions~\ref{assumption:noise}--\ref{assumption:subspace} to Section~\ref{sec:formal_results_discussion}. 

\subsection{Consistency Result} 
We now formally state our result characterizing the prediction accuracy of \SI-\PCR. 
For ease of notation, we absorb dependencies on $\sigma$ into the constant within $O_p(\cdot)$. 

\begin{theorem} \label{thm:consistency}
Given $(i,d)$, let Assumptions \ref{assumption:sutva} to \ref{assumption:subspace} hold. 
Consider the \SI-\PCR~estimator given by \eqref{eq:pcr} and suppose $k = r_{\epre} = \emph{rank}(\Ex[ \bYepred | \cE])$. 
Then conditioned on $\cE$, we have  
\begin{align}
	\htheta_i^{(d)} - \theta_i^{(d)}
	= O_p \left(\sqrt{\log(T_0 \numdonors)} \left[ \frac{r^{3/4}_\epre}{T_0^{1/4}} + r^{2}_\epre  \max\left\{\frac{\sqrt{\numdonors}}{T^{3/2}_0}, \frac{1}{\sqrt{T_0}}, \frac{1}{\sqrt{\numdonors}}\right\} \right] \right). 
\end{align}
Here, we define $O_p(\cdot)$ with respect to the sequence $\min\{N_d, T_0\}$.
\end{theorem}
Theorem~\ref{thm:consistency} establishes that the \SI-\PCR~estimator yields a consistent estimate of the causal estimand. 
More precisely, for a fixed $r_\pre$, the estimation error decays as $N_d$ and $T_0$ grow, provided $T_0 = \omega(N_d^{1/3})$. 
Notably, the duration of the post-treatment period, $T_1$, can be of fixed length. 
In fact, if $T_1 = 1$, then Theorem~\ref{thm:consistency} establishes {\em point-wise consistency}. 
We retain the flexibility of allowing $T_1 > 1$ since a common goal in practice is to estimate the average counterfactual outcome over the post-treatment period. 

\begin{remark}[Asymptotic normality] \label{remark:normality} 
{\em
Appendix~\ref{sec:proof.normality} provides conditions that lead to asymptotic normality, which then allows for the construction of valid confidence intervals. 
There, we show that a minor modification to the \SI-\PCR~estimator satisfies these conditions. 
}
\end{remark} 

\subsection{Remarks on Additional Assumptions} \label{sec:formal_results_discussion}

%

\begin{remark} [On Assumption~\ref{assumption:noise}]
{\em
While the latent time-treatment factors, $u_t^{(d)}$, can be arbitrarily correlated, Assumption~\ref{assumption:noise} enforces the idiosyncratic shocks to be independent. 
This can be restrictive. 
However, just as \cite{abadie2} used this independence structure to provide the first analysis of the \SC~estimator, we also adopt this independence structure to provide the first analysis of the \SI~estimator. 
A formal analysis of the \SI~estimator under more general noise models is left as important future work. 
}
\end{remark} 

\begin{remark} [On Assumption~\ref{assumption:boundedness}]
{\em 
The precise bound $[-1,1]$ is without loss of generality.
That is, it can be extended to $[a, b]$ for $a, b \in \Rb$ with $a \le b$.
}
\end{remark} 

\begin{remark} [On Assumption~\ref{assumption:spectra}]
{\em 
Assumption \ref{assumption:spectra} requires that the nonzero singular values of $\Ex[\bYpred | \Ec]$ are well-balanced, which we can empirically inspect through the spectral profile of $\bYpred$. 
Within the econometrics factor model and matrix completion literatures, Assumption \ref{assumption:spectra} is analogous to incoherence-style conditions \cite[Assumption A]{bai2020matrix} and notions of pervasiveness \cite[Proposition 3.2]{fan2018eigenvector}. 
Assumption \ref{assumption:spectra} has also been shown to hold w.h.p. 
for the canonical probabilistic generating process used to analyze probabilistic principal component 
analysis in \cite{bayesianpca} and \cite{probpca};
here, the observations are a high-dimensional embedding of a low-rank matrix with 
independent sub-Gaussian entries \cite[Proposition 4.2]{agarwal2020robustness}.
We highlight that the assumption of a gap between the top few singular values of observed matrix of interest and the remaining singular values has been widely adopted in the econometrics literature of large dimensional factor analysis dating back to \cite{chamberlainfactor}. 
}
\end{remark} 

\begin{remark} [On Assumption~\ref{assumption:subspace}]
{\em 
Assumption~\ref{assumption:subspace} is analogous to Assumption~\ref{assumption:linear}. 
As such, similar interpretations and implications carry over. 

Another interpretation of Assumption~\ref{assumption:subspace} views $\bYpred$ as the training set and $\bYpostd$ as the test set. 
Intuitively, generalization relies on a notion of ``similarity'' between the train and test sets. 
Standard arguments from the statistical learning literature formalizes this notion through the assumption that the two sets are constructed from i.i.d. draws. 
%
Since potential outcomes from different treatments are likely to arise from different distributions, such a postulation is ill-suited for our setting. 
Instead, Assumption~\ref{assumption:subspace} implies that generalization is achievable if the rowspace of $\Ex[\bYpostd | \cE]$ is contained within that of $\Ex[\bYpred| \cE]$ (see Lemma~\ref{lemma:identification.td} of Appendix~\ref{sec:proof.rowspaces}), i.e., each test point is a linear combination of the training set---a purely linear algebraic condition. 

Though $\Ex[\bYpostd| \cE]$ and $\Ex[\bYpred| \cE]$ are not observable, we can still investigate the validity of Assumption~\ref{assumption:spectra} using $\bYpostd$ and $\bYpred$ as proxies. 
To this end, we first recall that Assumption~\ref{assumption:linear} for unit $i$ is empirically tested through its pre-treatment fit, defined as 
\begin{align} \label{eq:pretreatment.fit}
	\rho_i = \frac{ \| (\bI - \bhU_\pre \bhU_\pre^\top) \Yipre \|_2} { \| \Yipre \|_2},
\end{align}
where $\bhU_\pre \in \Rb^{T_0 \times k}$ collects the top $k$ left singular vectors of $\bYpred$; note that the numerator of \eqref{eq:pretreatment.fit} measures the $\ell_2$-size of the in-sample residuals. 
Given the symmetry between Assumptions~\ref{assumption:linear} and \ref{assumption:subspace}, it is natural to define for every $t > T_0$, 
\begin{align}
	\phi_t = \frac{ \| (\bI - \bhV_\pre \bhV_\pre^\top) Y_{t, \Ic^{(d)}} \|_2} { \| Y_{t, \Ic^{(d)}} \|_2}, 
\end{align}
where $\bhV_\pre \in \Rb^{N_d \times k}$ collects the top $k$ right singular vectors of $\bYpred$ and $Y_{t, \Ic^{(d)}} = [Y_{tjd}: j \in \Ic^{(d)}]$. 
Observe that $\rho_i$ and $\phi_t$ largely depend on the ``richness'' of the row and column spaces of $\bYpred$, respectively. 
The row space, generated by $\Ic^{(d)}$, determines the existence of a synthetic unit $i$ formed as a weighted average of $\Ic^{(d)}$; the column space, generated by the pre-treatment outcomes (training set), controls the \SI-\PCR~estimator's ability to generalize to the post-treatment period (test set). 
By construction, $\rho_i$ and $\phi_t$ take values in $[0,1]$ with lower values being more desirable but higher values being more informative: a value close to zero does not guarantee the tested assumption holds but values close to one suggest that the tested assumption is likely violated. 
Thus, $\rho_i$ and $\phi_t$ are useful one-sided tests for the suitability of the \SI~framework toward a study. 

The concept of generalization under a change in distribution between the training and test sets is known in the machine learning community as {\em covariate shift}.
We hope that those exploring this area of research will find resonance with Assumption~\ref{assumption:subspace} and its implications, which Section~\ref{sec:simulations} explores in greater depth.
}
\end{remark} 


%% file: content/simulations_tldr.tex
\section{Simulation Study}\label{sec:simulations}
This section presents a simulation on the \SI-\PCR~estimator that complements the statistical claims of Theorem~\ref{thm:consistency}. 
In particular, we evaluate the {\em point-wise} prediction accuracy of the \SI-\PCR~estimator and study the role of Assumption~\ref{assumption:subspace}. 

\subsection{Data Generating Process} \label{sec:sim.dgp} 
We set $T_1 = 1$ and vary $N_d = T_0 \in \{25, 50, 75, \dots, 200\}$.  
We set $r =15$ and $r_\pre = 10$. 

\subsubsection{Latent factors} 
We generate the latent factors associated the units in $\Ic^{(d)}$, denoted as $\bV_{\cI^{(d)}} \in \Rb^{N_d \times r}$, by independently sampling its entries from a standard Normal distribution. 
Then, we sample $w^{(i,d)} \in \Rb^\numdonors$ by drawing its entries from a uniform distribution over $[0,1]$ and normalizing it to have unit norm. 
We define the target unit latent factor as $v_i = \bV_{\cI^{(d)}}^\top w^{(i,d)}$. 

Define the latent factors for the pre-treatment outcomes under control as $\bU^{(0)}_\pre = \bA \bB^\top$, where the entries of $\bA \in \Rb^{T_0 \times r_\pre}$ and $\bB \in \Rb^{r \times r_\pre}$ are i.i.d. samples from a standard Normal. 
Note $\rank(\bU^{(0)}_\pre) = r_\pre$. 
Next, we sample $\phi \in \Rb^r$ whose entries are i.i.d. draws from a uniform distribution over $[0, 1]$. 
Consider two post-treatment latent factors under treatment $d$:  
$u_\post^{(d, \cmark)} = (\bU_\pre^{(0)})^\dagger \bU^{(0)}_\pre \phi$ such that Assumption~\ref{assumption:subspace} holds ($\cmark$) with respect to $\bU^{(0)}_\pre$
and $u_\post^{(d, \xmark)} = [ \bI -  (\bU_\pre^{(0)})^\dagger \bU^{(0)}_\pre ] \phi$ such that Assumption~\ref{assumption:subspace} fails ($\xmark$) with respect to $\bU^{(0)}_\pre$. 
Thus, we define two causal estimands,  
$\theta_i^{(d, \cmark)} = \langle u_\post^{(d, \cmark)}, v_i \rangle$
and
$\theta_i^{(d, \xmark)} = \langle u_\post^{(d, \xmark)}, v_i \rangle$.  

\subsubsection{Observations}
Let
(i) $\Yipre = \bU^{(0)}_\pre v_i + \varepsilon_{\pre, i}$; 
(ii) $\bYpred =\bU^{(0)}_\pre \bV^\top_{\cI^{(d)}} + \bE_{\pre, \Ic^{(d)}}$; 
(iii) $\Ypostd^{(\cmark)} = \bV_{\cI^{(d)}} u_\post^{(d, \cmark)} + \varepsilon_{\post, \Ic^{(d)}}$; 
and
(iv) $\Ypostd^{(\xmark)} = \bV_{\cI^{(d)}} u_\post^{(d, \xmark)} + \varepsilon_{\post, \Ic^{(d)}}$, 
where the entries of $\varepsilon_{\pre, i}$, $\bE_{\pre, \Ic^{(d)}}$, and $\varepsilon_{\post, \Ic^{(d)}}$ are i.i.d. samples from a standard Normal. 

\subsection{Simulation Results} 
From  $\Ynpre$ and $\bYpred$, we learn a {\em single} linear model $\hw^{(i,d)}$ via the \SI-\PCR~estimator with the hyper-parameter $k$ chosen via the approach prescribed in \cite{Gavish_2014}. 
We then produce two estimates,
$\htheta_i^{(d, \cmark)} = \langle \Ypostd^{(\cmark)}, \hw^{(i,d)} \rangle$ 
and
$\htheta_i^{(d, \xmark)} = \langle \Ypostd^{(\xmark)}, \hw^{(i,d)} \rangle$. 
Figure~\ref{fig:consistency} visualizes the point-wise prediction errors of $| \htheta_i^{(d, \cmark)} - \theta_i^{(d, \cmark)}|$ and $| \htheta_i^{(d, \xmark)} - \theta_i^{(d, \xmark)}|$ across different dimensions and averaged over $50$ trials, where each trial consists of an independent draw of the latent factors and the observations of size $100$; this amounts to $5000$ total simulation repeats per dimension. 
We observe that the $\cmark$ point-wise errors decay with increasing dimensions, which corroborates with Theorem~\ref{thm:consistency}. 
The $\xmark$ point-wise errors fail to converge, which suggests that Assumption~\ref{assumption:subspace} can be critical for generalization. 

\begin{figure}[!t]
	\centering 
		\centering 
		\includegraphics[width=0.4\linewidth]
		{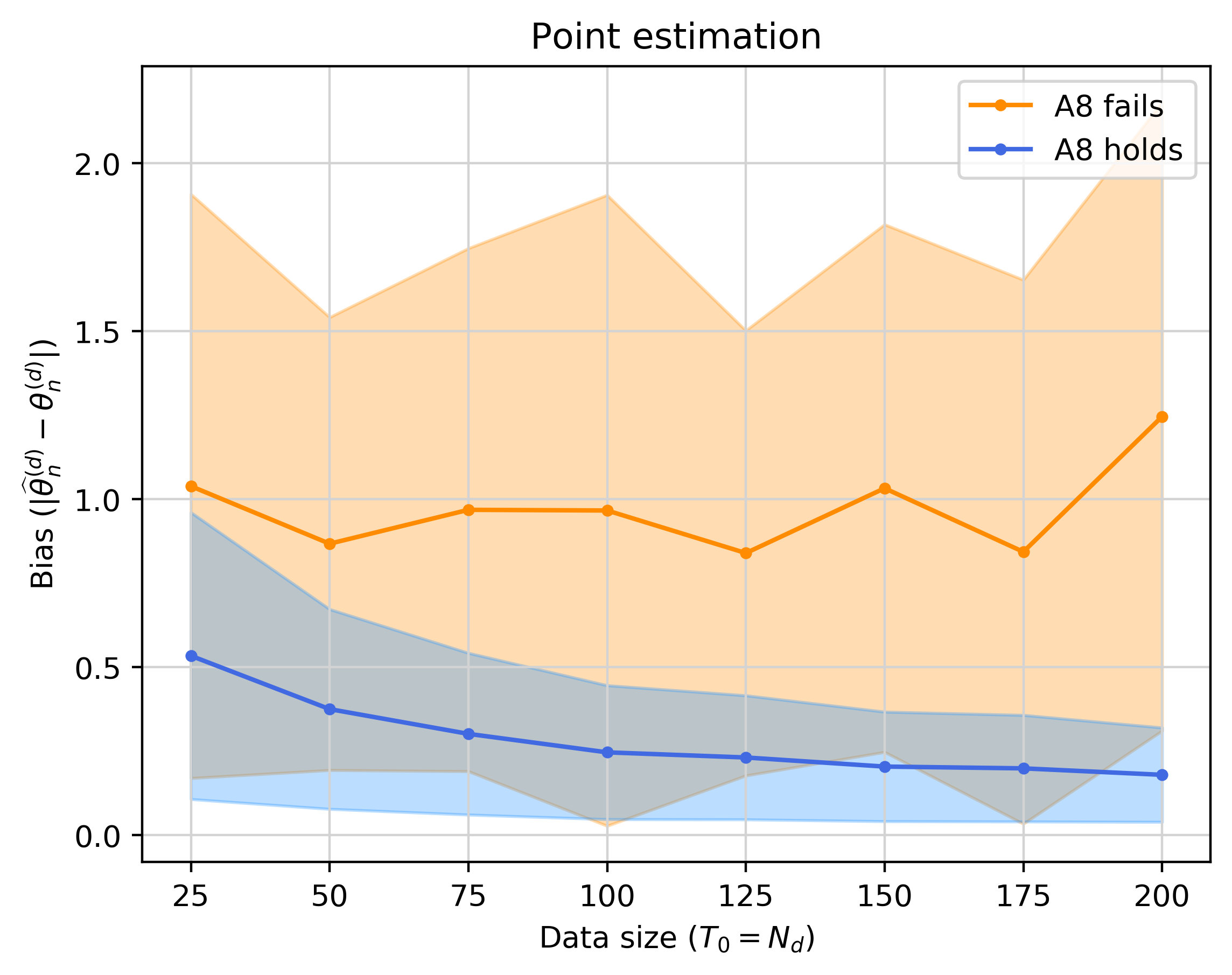}
	\caption{Simulation results displaying the absolute point-wise prediction errors across $T_0 = N_d \in \{25, 50, 75, \dots, 200\}$. Solid lines show the mean over $50$ trials, with shading to show $\pm$ one standard error. Blue and orange lines show the errors when Assumption~\ref{assumption:subspace} holds ($\cmark$) and fails ($\xmark$), respectively. In line with Theorem~\ref{thm:consistency}, the $\cmark$ errors decay as $T_0$ and $N_d$ grow. By contrast, the $\xmark$ errors fail to converge, which provides empirical evidence for the importance of Assumption~\ref{assumption:subspace}.}
	\label{fig:consistency} 
\end{figure}

%% file: content/empirical.tex
\section{Empirical Application}\label{sec:empirics} 
We revisit the classical Proposition 99 study of \cite{abadie2} introduced in Section~\ref{sec:intro}. 
Recall that the focus of \cite{abadie2} was to answer Q1. 
By contrast, our attention rests with Q2-Q5 given our interest in the \SI-\PCR~estimator when its linear model is learned on outcomes observed under the status quo and applied on outcomes observed under an anti-tobacco program or raised taxes.
More formally, whereas \cite{abadie2} aimed to estimate $\theta_i^{(0)}$, we set out to estimate $\theta_i^{(d)}$ for $d \in \{1,2\}$. 


\subsection{Setup} 
While we would ideally like to respond to Q2-Q5 directly, the answers to these questions are counterfactual, which precludes any evaluation. 
Accordingly, we emulate these questions via a leave-one-out (LOO) analysis. 
Specifically, for every treatment $d$, we iteratively designate a state $i \in \Ic^{(d)}$ as the target unit. 
We withhold its post-treatment outcomes, treating them as the ground truth, and define our causal estimand as $\theta_i^{(d)} = (1/T_1) \sum_{t > T_0} Y_{tid}$, acknowledging that our estimand is imperfect since it inherently incorporates the idiosyncratic shocks. 
We retain access to the target unit's pre-treatment outcomes, $Y_{ti0}$ for $t \le T_0$, as well as the outcomes for the remaining units in $\Ic^{(d)}$ for the entire time horizon, i.e., for all $j \in \Ic^{(d)} \setminus \{i\}$, we observe $Y_{tj0}$ for $t \le T_0$ and $Y_{tjd}$ for $t > T_0$. 
Our objective is to estimate the causal estimand from these observations with the \SI-\PCR~estimator. 

For each treatment $d$, we summarize the performance of the \SI-\PCR~estimator as $\texttt{errors}(d) = [| (\htheta_i^{(d)} - \theta_i^{(d)}) / \theta_i^{(d)}|: i \in \Ic^{(d)}]$, which collects the absolute normalized LOO prediction errors across $\Ic^{(d)}$. 
%
Practically speaking, $\texttt{errors}(d)$ measures the \SI-\PCR~estimator's ability to recreate the post-treatment outcomes under treatment $d$ using a linear model learned from pre-treatment outcomes under control (status quo). 
With this interpretation and given the widespread acceptance of the \SC~framework---particularly, in the context of this dataset---we view $\texttt{errors}(0)$ as the baseline; hence, the difference between $\texttt{errors}(0)$ and $\texttt{errors}(d)$ for $d \in \{1,2\}$ serves as a measure for the validity of the \SI~framework. 
In what follows, the hyper-parameter $k$ of the \SI-\PCR~estimator is chosen as the minimum number of singular values needed to capture $99\%$ of the spectral energy in $\bYpred$.


\subsection{Empirical Results} 
Table~\ref{table:errors} reports on the average and standard deviation of $\texttt{errors}(d)$ across all treatments $d \in \{0,1,2\}$.
Compared to the baseline, the average error and standard deviation associated with $d=2$ (raised taxes) is lower. 
At the same time, the average error for $d=1$ (anti-tobacco program) matches that of the baseline but its standard deviation is noticeably larger. 
This may be an artifact of there being significantly less states that imposed an anti-tobacco program ($N_1=5$) compared to keeping the status quo ($N_0 = 38$). 
However, if we redefine the treatments to take on two values as opposed to three, where we maintain the status quo as control ($d'=0$) but combine anti-tobacco programs and raised taxes into a single anti-tobacco measure ($d'=1$), then we find the resulting error summary (not reported in Table~\ref{table:errors}) to be $0.077 \pm 0.079$.  
Figure~\ref{fig:empirics} provides a visualization of the estimated trajectory of potential cigarette sales from 1989--2000 under the statewide status quo ($d=0$), an anti-tobacco program ($d=1$), and raised taxes ($d=2$) for three example states. 

\begin{table} [h]
    \centering
    \caption{Summary of the average and standard deviation for $\texttt{errors}(d)$ across all $d \in \{0,1,2\}$.}
    \label{table:errors}
    \begin{adjustbox}{max width=\textwidth}
    \begin{tabular}{l c  }
        \toprule
        Treatment  				&   Prediction Error   \\
        \midrule
        $d=0$: status quo     			&  $0.105 \pm 0.064$   		\\
        $d=1$: anti-tobacco program     &  $0.105 \pm 0.116$        		\\
        $d=2$: raised taxes 			&  $0.070 \pm 0.052$      		\\
        \bottomrule
    \end{tabular}
    \end{adjustbox}
\end{table}

Collectively, our results provide empirical evidence that the \SI~methodology, like the \SC~methodology, is suitable for this case study. 
Moreover, our findings suggest that (i) anti-tobacco programs and raised taxes yielded similar effects with (ii) both leading to a significant decrease in tobacco consumption compared to the status quo, which is consistent with the conclusions drawn in \cite{abadie2} for California. 

\begin{figure}[t]
	\centering 
	\begin{subfigure}[b]{0.325\textwidth}
		\centering 
		\includegraphics[width=\linewidth]{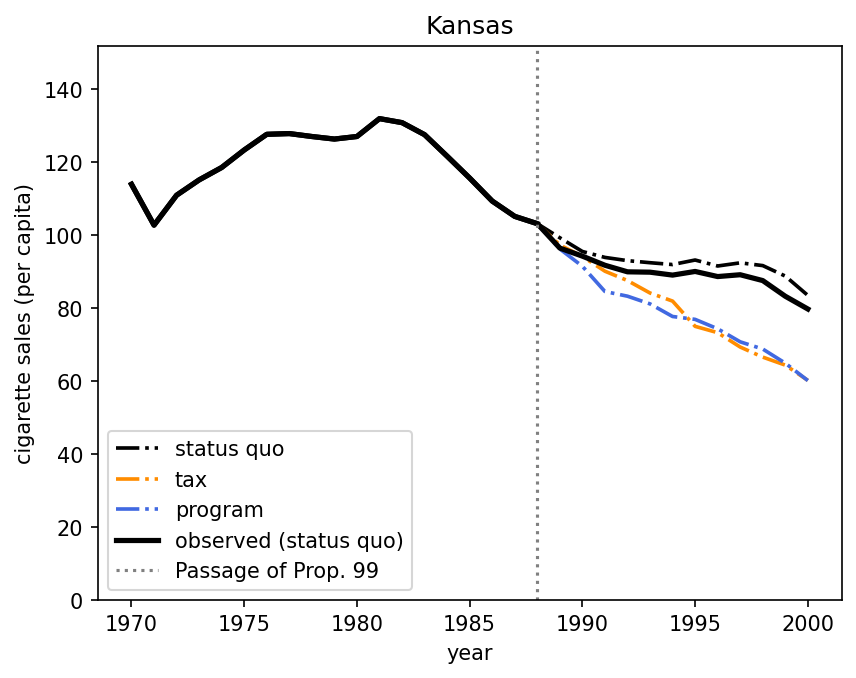}
		\caption{Kansas.} 
		\label{fig:kansas} 
	\end{subfigure} 
	\begin{subfigure}[b]{0.325\textwidth}
		\centering 
		\includegraphics[width=\linewidth]{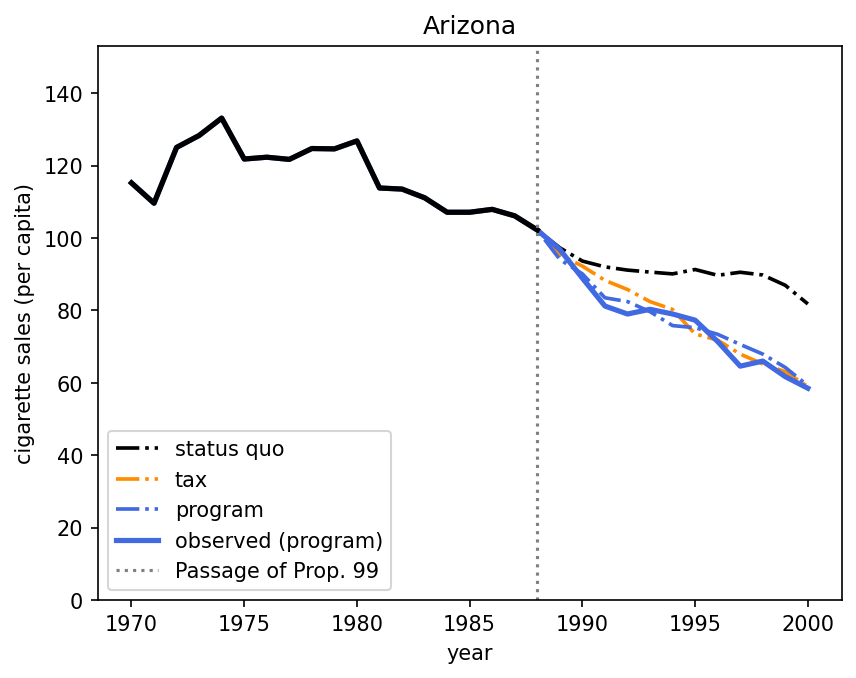}
		\caption{Arizona.} 
		\label{fig:arizona} 
	\end{subfigure} 
	\begin{subfigure}[b]{0.325\textwidth}
		\centering 
		\includegraphics[width=\linewidth]{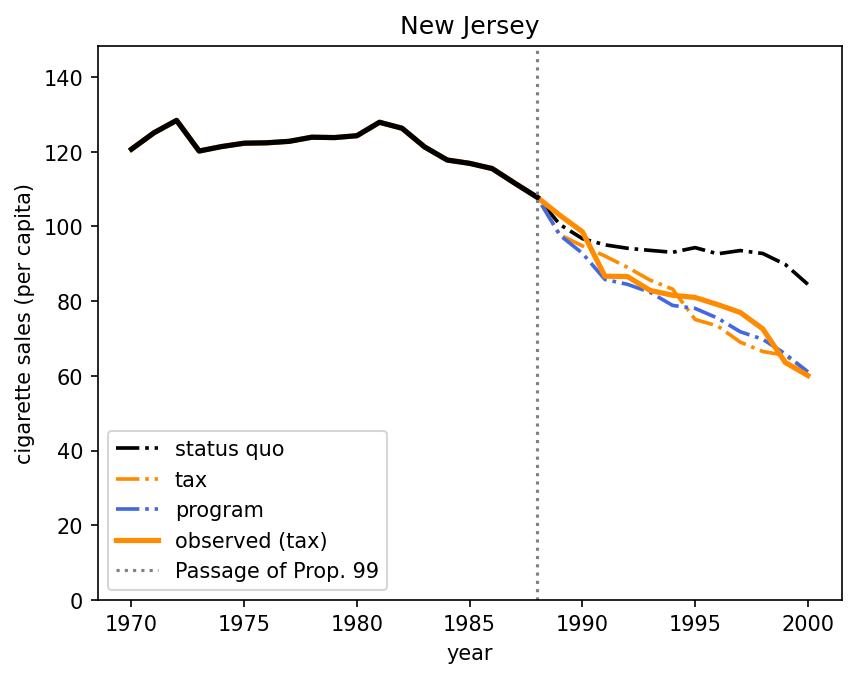}
		\caption{New Jersey.} 
		\label{fig:jersey} 
	\end{subfigure} 
	\caption{The solid and dashed-dotted lines shows the observed and estimated trajectory of cigarette sales, respectively. 
	Treatments are indexed by color: status quo in black, anti-tobacco program in blue, and raised taxes in orange.  
	The vertical dotted gray line represents the year 1988 when California passed Proposition 99; this splits the time horizon into the pre-treatment period (1970--1988) and post-treatment period (1989--2000). 
	Figures~\ref{fig:kansas}, \ref{fig:arizona}, and \ref{fig:jersey} show the results for an example state that in actuality kept the status quo (Kansas), imposed an anti-tobacco program (Arizona), and raised taxes (New Jersey), respectively. 
	}
	\label{fig:empirics} 
\end{figure}

%% file: content/conclusion.tex
\section{Conclusion} \label{sec:conclusion}
In this article, we reinterpret the problem of counterfactual estimation as one of tensor completion. 
To enable faithful recovery of the underlying tensor, we anchor on the low-rank tensor factor model, which is arguably the de-facto assumption in the tensor completion literature. 
The tensor factor model is a natural generalization of the canonical matrix factor model used to justify the \SC~framework with the added assumption of a latent factorization over treatments and shared latent unit factors across not only time (as in the matrix factor model) but treatments as well. 
Accordingly, we propose the \SI~framework, which rests upon this structural assumption. 
As a consequence, our proposed framework extends the \SC~framework to the setting of multiple treatments and, therefore, provides one resolution to the open question posed in \cite{abadie_survey}. 

We couple our framework with an estimation strategy that generalizes the standard \SC-based estimators, and show that one variation of our approach, \SI-\PCR, is a consistent estimator of our causal estimand of interest. 
Our statistical claims are supplemented by a simulation study that highlights the substantive implications of a linear algebraic condition (Assumption~\ref{assumption:subspace}) for the prediction accuracy of the \SI-\PCR~variant. 
Finally, our empirical study provides evidence that the \SI~methodology, like the \SC~methodology, may be suitable for the classical tobacco study of \cite{abadie2}. 
Complementing the conclusions of \cite{abadie2}, our findings suggest that anti-tobacco programs and raised taxes had similar effects with both leading to a significant reduction in tobacco consumption compared to the status quo. 

We envision several interesting avenues of future research. 
Keeping with our tensor factor model, an analysis of the \SI-\PCR~estimator under less restrictive assumptions on the noise model (Assumption~\ref{assumption:noise}) would offer deeper insights into its predictive behavior for more general settings.
More broadly, formal analyses on different formulations of the \SI~estimator may provide guidance to researchers on the choice of formulation for their study, should the \SI~framework be suitable.
In light of the recent work of \cite{bottmer} that provided the first design-based analysis of the \SC~method, it would also be fascinating to understand the design-based properties of the the \SI~method to study its utility for randomized experiments. 

%


%% file: content/proofs_identification.tex
\section{Proof of Theorem \ref{thm:identification}} \label{sec:identification} 
\begin{proof} 
We have that 
\begin{align}
\Ex [ Y^{(d)}_{ti} \big| u^{(d)}_t, v_i  ]  
& = \Ex [\langle u^{(d)}_t,  v_i \rangle + \varepsilon^{(d)}_{ti} ~\big|~ u^{(d)}_t, v_i ] &&\because \text{Assumption~\ref{assumption:form}}
\\& = \langle u^{(d)}_t,  v_i \rangle ~\big|~ \{u^{(d)}_t, v_i\} &&\because \text{Assumption~\ref{assumption:conditional_mean_zero}} 
\\&=  \langle u^{(d)}_t,  v_i \rangle ~\big|~ \cE 
\\&= \langle u^{(d)}_t, \sum_{j \in \Ic^{(d)}} w_j^{(i,d)} v_j \rangle ~\big|~ \cE &&\because \text{Assumption~\ref{assumption:linear}} 
\\& = \sum_{j \in \Ic^{(d)}} w_j^{(i,d)} \cdot
\Ex \left[ \langle u^{(d)}_t, v_j \rangle + \varepsilon^{(d)}_{tj}  ~\big|~ \cE \right] &&\because \text{Assumption~\ref{assumption:conditional_mean_zero}} 
\\ & = \sum_{j \in \Ic^{(d)}} w_j^{(i,d)} \cdot \Ex [  Y^{(d)}_{tj} \big| \cE ] &&\because \text{Assumption~\ref{assumption:form}} 
\\& = \sum_{j \in \Ic^{(d)}} w_j^{(i,d)} \cdot \Ex \left[  Y_{tjd} | \cE \right]. &&\because \text{Assumption~\ref{assumption:sutva}} %
\end{align}
%
%
The third equality follows since $\langle u^{(d)}_t, v_i \rangle$ is deterministic conditioned on $\{u^{(d)}_t, v_i\}$.
Recalling \eqref{eq:causal.est},
one can verify \eqref{eq:identification} using the same the argument used to prove \eqref{eq:identification_strong} but conditioning on the set of latent factors $\{u^{(d)}_t, v_i: t > T_0 \}$ at the start rather than just $\{u^{(d)}_t, v_i\}$. 
\end{proof} 

%% file: content/proofs_interpretations.tex
\section{Helper Lemmas to Prove Theorem~\ref{thm:consistency}} \label{sec:helper.lemmas} 
In this section, we state and prove several helper lemmas to establish Theorem~\ref{thm:consistency}. 
All $O_p(\cdot)$ statements are with respect to the sequence $\min\{T_0, \numdonors\}$.

\subsection{Perturbation of Singular Values} \label{sec:proof_singular_values} 
In Lemma \ref{lemma:singular_values}, we argue that the singular values of $\bYpred$ and $\Ex[\bYpred | \cE]$ are close. 
Recall that $s_\ell$ and $\hs_\ell$ are the singular values of $\Ex[\bYpred | \cE]$ and $\bYpred$, respectively. 

\begin{lemma} \label{lemma:singular_values}
Let Assumptions \ref{assumption:sutva},  \ref{assumption:form}, \ref{assumption:conditional_mean_zero},  
\ref{assumption:noise}, \ref{assumption:boundedness}
and \ref{assumption:spectra} hold. 
Conditioned on $\Ec$, we have for any $d$, $\zeta > 0$, and $\ell \le \min\{T_0, \numdonors\}$, 
\begin{align}
	\abs{s_\ell - \hs_\ell} \le C\sigma \left( \sqrt{T_0} + \sqrt{\numdonors} + \zeta \right)
\end{align} 
w.p. at least $1-2\exp(-\zeta^2)$, where $C>0$ is an absolute constant. 
\end{lemma}

\begin{proof} {Proof of Lemma~\ref{lemma:singular_values}.} 
For ease of notation, we suppress the conditioning on $\Ec$ for the remainder of the proof.
To bound the gap between $s_\ell$ and $\hs_\ell$, we recall the following well-known results.  

\begin{lemma}  [Weyl's inequality] \label{lemma:weyl}
Given $\bA, \bB \in \Rb^{m \times n}$, let $\sigma_\ell$ and $\widehat{\sigma}_\ell$ be the $\ell$-th singular values of $\bA$ and $\bB$, respectively, in decreasing order and repeated by multiplicities. 
Then for all $i \le \min\{m, n\}$, $\abs{ \sigma_\ell - \widehat{\sigma}_\ell} \le \|\bA - \bB \|_\eop.$
\end{lemma} 

\begin{lemma} [Sub-Gaussian Matrices: Theorem 4.4.5 of \cite{vershynin2018high}] \label{lemma:subg_matrix}
Let $\bA = [A_{ij}]$ be an $m \times n$ random matrix where the entries $A_{ij}$ are independent, mean zero, sub-Gaussian random variables. 
Then for any $t > 0$, we have $\| \bA \|_{\eop} \le CK \left(\sqrt{m} + \sqrt{n} + t \right)$ w.p. at least $1-2\exp(-t^2)$. 
Here, $K = \max_{i,j} \|A_{ij}\|_{\psi_2}$ and $C>0$ is an absolute constant. 
\end{lemma} 

By Lemma \ref{lemma:weyl}, we have for any $\ell \le \min\{T_0, N_d\}$, $\abs{s_\ell - \hs_\ell}  \leq \| \bYpred - \Ex[\bYpred]\|_\txtop$. 
Recalling Assumption \ref{assumption:noise} and applying Lemma \ref{lemma:subg_matrix}, we conclude for any $t>0$ and some absolute constant $C>0$, $\abs{s_\ell - \hs_\ell} \le C\sigma\left(\sqrt{T_0} + \sqrt{N_d} + t \right)$ w.p. at least $1-2\exp(-t^2)$. 
This completes the proof.
\end{proof}

\subsection{On the Rowspaces of the Pre- and Post-Intervention Expected Outcomes} \label{sec:proof.rowspaces} 
Similar to Assumption~\ref{assumption:linear}, observe that Assumption~\ref{assumption:subspace} implies that there exists a $\beta^{(d)} \in \Rb^{T_0}$ such that 
\begin{align}
	u_t^{(d)} = \sum_{\tau \le T_0} \beta_\tau^{(d)} \cdot u^{(0)}_\tau. \label{eq:subspace.beta} 
\end{align}
With this interpretation, we present a similar result to Theorem~\ref{thm:identification} that relates the rowspaces between the pre- and post-intervention expected outcomes. 
\begin{lemma} \label{lemma:identification.td}
Given $d$, let Assumptions~\ref{assumption:sutva}, \ref{assumption:form}, \ref{assumption:conditional_mean_zero}, and \ref{assumption:subspace} hold. 
Then given $\beta^{(d)}$ as defined in \eqref{eq:subspace.beta}, we have for all $t > T_0$ and $j \in [N]$,  
\begin{align}
	\Ex [ Y_{tjd} | \cE] &= \sum_{\tau \le T_0} \beta_\tau^{(d)} \cdot \Ex[ Y_{\tau j 0}  | \cE]. 
\end{align}
\end{lemma} 
\begin{proof} 
We have that
\begin{align}
	\Ex[Y_{tjd} | \cE] &= \Ex[Y^{(d)}_{tj} | \cE] && \because \text{Assumption~\ref{assumption:sutva}} 
	\\ &= \Ex[ \langle u_t^{(d)}, v_j \rangle + \varepsilon_{tj}^{(d)} | \cE] && \because \text{Assumption~\ref{assumption:form}}
	\\ &= \langle u_t^{(d)}, v_j \rangle | \cE && \because \text{Assumption~\ref{assumption:conditional_mean_zero}}  
	\\ &= \langle \sum_{\tau \le T_0} \beta^{(d)}_\tau \cdot u_\tau^{(0)}, v_j \rangle | \cE && \because \text{Assumption~\ref{assumption:subspace}} 
	\\ &= \sum_{\tau \le T_0} \beta^{(d)}_\tau \cdot \langle  u_\tau^{(0)}, v_j \rangle | \cE
	\\ &= \sum_{\tau \le T_0} \beta^{(d)}_\tau \cdot \Ex[ \langle  u_\tau^{(0)}, v_j \rangle + \varepsilon_{\tau j}^{(0)} | \cE ] && \because \text{Assumption~\ref{assumption:conditional_mean_zero}}
	\\ &= \sum_{\tau \le T_0} \beta^{(d)}_\tau \cdot \Ex[Y^{(0)}_{\tau j} | \cE] && \because \text{Assumption~\ref{assumption:form}}
	\\ &= \sum_{\tau \le T_0} \beta^{(d)}_\tau \cdot \Ex[Y_{\tau j 0} | \cE]. && \because \text{Assumption~\ref{assumption:sutva}} 
\end{align} 
This completes the proof. 
\end{proof}

\subsection{Rewriting Theorem~\ref{thm:identification} with the Minimum $\ell_2$-Norm Solution} 
Let $\bV_\pre \in \Rb^{\numdonors \times r_{\pre}}$ and $\bV_\post \in \Rb^{\numdonors \times r_{\post}}$ denote the right singular vectors of $\Ex[\bYpred | \cE ]$ and $\Ex[\bYpostd | \cE ]$, respectively. 
The next result rewrites Theorem~\ref{thm:identification} via the following: 
\begin{align} \label{eq:wtilde}
\tw^{(i,d)} = \bV_\pre \bV_\pre^\top w^{(i,d)}. 
\end{align} 
In words, $\tw^{(i,d)}$ is the orthogonal projection of $w^{(i,d)}$ onto the rowspace of $\Ex[ \bYpred | \cE]$ and is the minimum $\ell_2$-norm solution to \eqref{eq:identification_strong} and \eqref{eq:identification} of Theorem \ref{thm:identification}.

\begin{lemma} \label{lemma:theta} 
Given $(i, d)$, 
let the setup of Theorem \ref{thm:identification} and Assumption \ref{assumption:subspace} hold. %
Then, we have 
\begin{align}
\theta^{(d)}_i = \frac{1}{T_1} \sum_{t > T_0} \sum_{j \in \Ic^{(d)}} 
        \tw_j^{(i,d)} \cdot \Ex[Y_{tjd} | \Ec].
\end{align}
\end{lemma}
\begin{proof} 
From Lemma~\ref{lemma:identification.td}, it follows that the rowspan of $\Ex[\bYpostd | \cE]$ lies within that of $\Ex[\bYpred | \cE]$. 
Hence,  
\begin{align}\label{eq:subspace_inclusion_implication}
    \bV_\post &= \bV_\pre \bV_\pre^\top \bV_\post.
\end{align}
Therefore, we have 
\begin{align} \label{eq:theta.0} 
    \Ex[\bYpostd | \cE] \cdot \tw^{(i,d)} 
    &=  \Ex[\bYpostd | \cE]  \bV_\pre \bV_\pre^\top \cdot w^{(i,d)} \nonumber \\
 & = \Ex[\bYpostd | \cE] \cdot w^{(i,d)}, 
\end{align}
where we use \eqref{eq:subspace_inclusion_implication} in the last equality.
Hence, we conclude
\begin{align}
    \theta^{(d)}_i 
    &= \frac{1}{T_1} \sum_{t > T_0} \sum_{j \in \Ic^{(d)}} w^{(i,d)}_j \cdot \Ex[Y_{tjd} | \cE ] 
    = \frac{1}{T_1} \sum_{t > T_0} \sum_{j \in \Ic^{(d)}} \tw^{(i,d)}_j \cdot \Ex[Y_{tjd} | \cE ]. 
\end{align}
The first equality follows from \eqref{eq:identification} in Theorem \ref{thm:identification} and the second equality follows from \eqref{eq:theta.0}.
\end{proof}

\subsection{Model Identification}
\begin{lemma} [Corollary 4.1 of \cite{pcr_aos}]  \label{lemma:param_est} 
Given $(i,d)$, 
let Assumptions \ref{assumption:sutva} to \ref{assumption:subspace} hold.
Further, suppose $k = r_\epre = \rank(\Ex[\bYepred \, | \, \cE])$, where $k$ is defined as in \eqref{eq:pcr}.
Then conditioned on $\Ec$, we have  
\begin{align} \label{eq:param_est_lem} 
\hw^{(i,d)} - \tw^{(i,d)}  &=
O_p\left(\sqrt{\log(T_0 \numdonors)} \left[ \frac{r^{3/4}_\epre }{T^{1/4}_0 \numdonors^{1/2}} + \frac{r^{3/2}_\epre }{\min\{T_0,\numdonors\}} \right] \right).
\end{align}  
\end{lemma} 

\begin{proof} 
We first  restate \cite[Corollary 4.1]{pcr_aos}. 
\begin{lemma}[Corollary 4.1 of \cite{pcr_aos}]
Let the setup of Lemma \ref{lemma:param_est} hold. 
Then w.p. at least $1 - O(1 / (T_0 \numdonors)^{10})$,
\begin{align}\label{eq:corollary_iew_PCR}
\|\hw^{(i,d)} -  \tw^{(i,d)}\|^2_2 \le C(\sigma) \cdot \log(T_0 \numdonors) \left( \frac{r^{3/2}_\epre }{T^{1/2}_0 \numdonors} + \frac{r^{3}_\epre }{\min\{T_0,\numdonors\}^2}\right),
%
\end{align}
where $C(\sigma)$ is a constant that only depends on $\sigma$.
\end{lemma}
The proof follows by adapting the notation in \cite{pcr_aos} to that used in this article.
In particular, $y = Y_{\pre, i}, \bX = \Ex[\bYpred | \cE], \widetilde{\bZ} = \bYpred, \hbeta = \hw^{(i,d)}, \beta^* = \tw^{(i,d)}, T_0 = n, \numdonors = p, \rho = 1, d = 1$, where $y, \bX, \widetilde{\bZ}, \hbeta, \beta^*, n, p, \rho, d$ are the notations used in 
\cite{pcr_aos};
since we do not consider missing values in this work, $\widetilde{\bZ} = \bZ$, where $\bZ$ is the notation used in \cite{pcr_aos}.
We also use the fact that $\bX \beta^*$ in the notation of \cite{pcr_aos}, i.e., $\Ex[\bYpred | \cE ] \cdot \tw^{(i,d)}$ in our notation, 
equals $\Ex[Y_{\pre, i} | \cE]$.
This follows from
\begin{align}\label{eq:pre_intervention_identification}
    \Ex[Y_{\pre, i} | \cE] = \Ex[\bYpred | \cE] \cdot \tw^{(i,d)},
\end{align}
which follows from \eqref{eq:identification_strong} in Theorem \ref{thm:identification} and \eqref{eq:theta.0}.
We conclude by noting \eqref{eq:corollary_iew_PCR} gives our desired result. 
\end{proof}

\begin{lemma}[Lemma 19 of \cite{synthetic_combinations}]\label{lemma:coeff_bound}
Given $(i,d)$, 
let Assumptions \ref{assumption:sutva} to \ref{assumption:subspace} hold.
Then conditioned on $\cE$, we have that 
$$\| \tw^{(i,d)}  \|_2  \le C \sqrt{\frac{r_\epre}{\numdonors}}$$
for some constant $C > 0$. 
This immediately implies that
$\| \tw^{(i,d)}  \|_1  \le C \sqrt{r_\epre}.$
\end{lemma}

\begin{proof} 
As in the proof of Lemma \ref{lemma:param_est}, the result immediately holds by matching notation to that of \cite{synthetic_combinations}.
\end{proof}

%

%% file: content/proofs_consistency_v2.tex
\section{Proof of Theorem \ref{thm:consistency}} \label{sec:proof_consistency} 
Armed with our helper lemmas in Section~\ref{sec:helper.lemmas}, we are ready to formally establish Theorem~\ref{thm:consistency}. 
We adopt the following notation. 
%
%
For any $t > T_0$,  let $Y_{t, \Ic^{(d)}} = [Y_{tjd}: j \in \Ic^{(d)}] \in \Rb^{\numdonors}$
and 
$\varepsilon_{t, \Ic^{(d)}} = [\varepsilon^{(d)}_{tj}: j \in \Ic^{(d)}] \in \Rb^{\numdonors}$. 
Note that the rows of $\bYpostd$ are formed by $\{Y_{t, \Ic^{(d)}}: t > T_0\}$. 
Additionally, let $\Delta^{(i,d)} = \hw^{(i,d)} - \tw^{(i,d)}$. 
Finally, for any matrix $\bA$ with orthonormal columns, let $\Pc_A = \bA \bA^\top$ denote the projection matrix onto the subspace spanned by the columns of $\bA$. 
We recall that all $O_p(\cdot)$ statements are with respect to the sequence $\min(T_0, \numdonors)$.

\medskip 
\begin{proof} 
For ease of notation, we suppress the conditioning on $\Ec$ throughout the proof.
By \eqref{eq:si.2} and Lemma \ref{lemma:theta} (implied by Theorem \ref{thm:identification}), we have 
\begin{align}
    \htheta^{(d)}_i - \theta^{(d)}_i 
    &= \frac{1}{T_1} \sum_{t > T_0} 
     \left( \langle Y_{t, \Ic^{(d)}}, \hw^{(i,d)} \rangle - \langle \Ex[ Y_{t, \Ic^{(d)}}], \tw^{(i,d)} \rangle \right) 
     \\ &= 
     \frac{1}{T_1} \sum_{t > T_0} \left( \langle \Ex[Y_{t, \Ic^{(d)}}], \Delta^{(i,d)} \rangle + \langle \varepsilon_{t, \Ic^{(d)}}, \tw^{(i,d)} \rangle
     + \langle \varepsilon_{t, \Ic^{(d)}},\Delta^{(i,d)} \rangle \right), \label{eq:inf.0} 
\end{align}
where we have used $Y_{t, \Ic^{(d)}} = \Ex[ Y_{t, \Ic^{(d)}}] + \varepsilon_{t, \Ic^{(d)}}$.  
From Lemma~\ref{lemma:identification.td}, it follows that $\bV_\post= \Pc_{V_\pre} \bV_\post$, 
where $\bV_\pre$ and $\bV_\post$ are the right singular vectors of $\Ex[\bYpred]$ and $\Ex[\bYpostd]$, respectively.
Hence, $\Ex[\bYpostd] = \Ex[\bYpostd] \Pc_{V_\pre}$. 
As such, for any $t > T_0$,
\begin{align}
    \langle \Ex[Y_{t, \Ic^{(d)}}], \Delta^{(i,d)} \rangle
    &= \langle \Ex[Y_{t, \Ic^{(d)}}],  \Pc_{V_\pre} \Delta^{(i,d)} \rangle. \label{eq:inf.01}
\end{align}
Plugging \eqref{eq:inf.01} into \eqref{eq:inf.0} yields
\begin{align}
\htheta^{(d)}_i - \theta^{(d)}_i &= \frac{1}{T_1}\sum_{t > T_0} \left(\langle \Ex[Y_{t, \Ic^{(d)}}],  \Pc_{V_\pre} \Delta^{(i,d)} \rangle + \langle \varepsilon_{t, \Ic^{(d)}}, \tw^{(i,d)} \rangle
+  \langle \varepsilon_{t, \Ic^{(d)}}, \Delta^{(i,d)} \rangle \right). \label{eq:inf.1} 
\end{align} 
Below, we bound the three terms on the right-hand side (RHS) of \eqref{eq:inf.1} separately. 

\medskip   
\noindent
{\em Bounding term 1.}  
By Cauchy-Schwartz inequality, observe that
$$
\langle \Ex[Y_{t, \Ic^{(d)}}],  \Pc_{V_\pre} \Delta^{(i,d)} \rangle 
\le 
\| \Ex[Y_{t, \Ic^{(d)}}] \|_2 \cdot \| \Pc_{V_\pre} \Delta^{(i,d)} \|_2. 
$$
Under Assumption \ref{assumption:boundedness}, we have
$\| \Ex[Y_{t, \Ic^{(d)}}] \|_2 \le \sqrt{N_d}$. 
As such, 
\begin{equation}\label{eq:problem_term}
  \frac{1}{T_1} \sum_{t > T_0}  \langle \Ex[Y_{t, \Ic^{(d)}}],  \Pc_{V_\pre} \Delta^{(i,d)} \rangle
    \le \sqrt{N_d} \cdot \| \Pc_{V_\pre} \Delta^{(i,d)} \|_2. 
\end{equation}
Hence, it remains to bound $\| \Pc_{V_\pre} \Delta^{(i,d)} \|_2$. 
Towards this, we state the following lemma.
Its proof can be found in Appendix \ref{sec:consistency.1}. 
\begin{lemma} \label{lemma:pcr_aos2}
Consider the setup of Theorem \ref{thm:consistency}.
Then,
$$
\Pc_{V_\epre} \Delta^{(i,d)} 
= O_p \left( \frac{\sqrt{r_\epre}}{\sqrt{N_d} T_0^{\frac 1 4}} + \frac{r^{3/2}_\epre  \sqrt{\log(T_0 \numdonors)}}{\sqrt{N_d} \cdot \min\{\sqrt{T_0}, \sqrt{N_d} \}}   
+  \frac{r^{2}_\epre \sqrt{\log(T_0 \numdonors)}}{\min\{T_0^{3/2},\numdonors^{3/2}\}}
\right).
$$
%
\end{lemma} 

Using Lemma \ref{lemma:pcr_aos2}, we obtain  
\begin{align}
& \frac{1}{T_1} \sum_{t > T_0}  \langle \Ex[Y_{t, \Ic^{(d)}}],  \Pc_{V_\pre} \Delta^{(i,d)} \rangle
\\ &= O_p \Bigg( \frac{r^{1/2}_\pre}{T_0^{\frac 1 4}} + \frac{r^{3/2}_\pre  \sqrt{\log(T_0 \numdonors)}}{\min\{\sqrt{T_0}, \sqrt{N_d} \}}  + \frac{r^{2}_\pre \sqrt{\numdonors \log(T_0 \numdonors)}}{\min\{T_0^{3/2},\numdonors^{3/2}\}} \Bigg). \qquad
\label{eq:consistency.1} 
\end{align}
This concludes the analysis for the first term. 

\medskip 
\noindent
{\em Bounding term 2. }
We begin with a simple lemma, which immediately holds by Hoeffding's Lemma (Lemma~\ref{lemma:hoeffding_random}). 
\begin{lemma}\label{lemmma:standard_convergences}
Let $\gamma_t$ be a sequence of independent mean zero sub-Gaussian random variables with variance less than $\sigma^2$. 
Then, for any vector $\boldsymbol{a} = (a_1, \dots, a_T) \in \Rb^T$, we have 
$
\sum^T_{t = 1} a_t \gamma_t = O_p\left( \sigma \| \boldsymbol{a} \|_2 \right).
$
\end{lemma}
%
%
By Assumptions \ref{assumption:form} and \ref{assumption:noise}, as well as Lemma \ref{lemmma:standard_convergences}, we have
\begin{align} 
     \sum_{t > T_0} \langle \varepsilon_{t, \Ic^{(d)}}, \tw^{(i,d)} \rangle 
    &= O_p \left( \sqrt{T_1} \| \tw^{(i,d)} \|_2 \right). 
\end{align}
Here, we have used the fact that the components of $\langle \varepsilon_{t, \Ic^{(d)}},  \tw^{(i,d)} \rangle$ are independent mean-zero sub-Gaussians and that $\varepsilon_{t, \Ic^{(d)}}$ are independent across $t$. 
Then using Lemma \ref{lemma:coeff_bound}, we have that 
\begin{align} \label{eq:consistency.2}
    \frac{1}{T_1}\sum_{t > T_0} \langle \varepsilon_{t, \Ic^{(d)}}, \tw^{(i,d)} \rangle 
    &= O_p \left( \sqrt{\frac{r_\pre}{\numdonors T_1}} \right). 
\end{align}

\medskip 
\noindent
{\em Bounding term 3. }
First, we define the event $\Ec_1$ as 
$$
\Ec_1 = \left\{ \|\Delta^{(i,d)}  \|_2 = O\left(\sqrt{\log(T_0 \numdonors)} \left[ \frac{r^{3/4}_\pre }{T^{1/4}_0 \numdonors^{1/2}} + \frac{r^{3/2}_\pre }{\min\{T_0,\numdonors\}} \right]  \right) \right\}. 
$$
By Lemma \ref{lemma:param_est}, $\Ec_1$ occurs w.h.p. 
Next, we define $\Ec_2$ as 
$$
\Ec_2 = \Bigg\{ 
    \sum_{t > T_0} \langle \varepsilon_{t, \Ic^{(d)}}, \Delta^{(i,d)} \rangle
    =  O\left( \sqrt{T_1 \log(T_0 \numdonors)} \left[ \frac{r^{3/4}_\pre }{T^{1/4}_0 \numdonors^{1/2}} + \frac{r^{3/2}_\pre }{\min\{T_0,\numdonors\}}\right] \right)
\Bigg\}. 
$$
Now, condition on $\Ec_1$. 
Since $\hw^{(i,d)}$ only depends on $\bYpred$ and $Y_{\pre, i}$, it is independent of $\varepsilon_{t, \Ic^{(d)}}$ for all $t > T_0$. 
Hence, Lemmas \ref{lemma:param_est} and \ref{lemmma:standard_convergences} imply $\Ec_2 | \Ec_1$ occurs w.h.p.
Further, we note 
\begin{align}\label{eq:tower_law}
\Pb(\Ec_2) = \Pb(\Ec_2 | \Ec_1) \Pb(\Ec_1) + \Pb(\Ec_2 | \Ec_1^c) \Pb(\Ec_1^c) \ge \Pb(\Ec_2 | \Ec_1) \Pb(\Ec_1). 
\end{align}
Since $\Ec_1$ and $\Ec_2 | \Ec_1$ occur w.h.p, it follows from \eqref{eq:tower_law} that $\Ec_2$ occurs w.h.p.
As a result,  
\begin{align}
    \frac{1}{T_1} \sum_{t > T_0} \langle \varepsilon_{t, \Ic^{(d)}}, \Delta^{(i,d)} \rangle 
    &= O_p \left( \frac{\sqrt{\log(T_0 \numdonors)}}{T^{1/2}_1} \left[ \frac{r^{3/4}_\pre }{T^{1/4}_0 \numdonors^{1/2}} + \frac{r^{3/2}_\pre }{\min\{T_0,\numdonors\}} \right]  \right). \label{eq:consistency.3}
\end{align}

\medskip \noindent 
{\em Collecting terms. }
Plugging \eqref{eq:consistency.1}, \eqref{eq:consistency.2}, \eqref{eq:consistency.3} into \eqref{eq:inf.1}, and simplifying
gives our desired result. 
%
%
\end{proof}

\subsection{Proof of Lemma \ref{lemma:pcr_aos2}} \label{sec:consistency.1} 
We introduce some helpful notation. 
Let $\bYpredrpre = \sum_{\ell=1}^{r_{\pre}} \hs_\ell \hu_\ell \hv^\top_\ell$ be the rank-$r_\pre$-approximation of $\bYpred$. 
More compactly, $\bYpredrpre = \bhU_\pre \bhS_\pre \bhV^\top_\pre$. 

\medskip 
\begin{proof} 
To establish Lemma \ref{lemma:pcr_aos2}, consider the following decomposition: 
$$
\Pc_{V_\pre} \Delta^{(i,d)} 
= (\Pc_{V_\pre} - \Pc_{\hV_\pre}) \Delta^{(i,d)}
+ \Pc_{\hV_\pre} \Delta^{(i,d)}. 
$$
We proceed to bound each term separately. 

\vspace{10pt} 
\noindent
{\em Bounding term 1. }
Recall $\| \bA v \|_2 \le \| \bA \|_\txtop \cdot \| v \|_2$ for any $\bA \in \Rb^{a \times b}$ and $v \in \Rb^b$. 
Thus, 
\begin{align}\label{eq:intermediate_Cauchy_Schwarz}
\| (\Pc_{V_\pre} - \Pc_{\hV_\pre}) \Delta^{(i,d)} \|_2 
\le \| \Pc_{V_\pre} - \Pc_{\hV_\pre}\|_\txtop \cdot \| \Delta^{(i,d)} \|_2.  
\end{align}
To control \eqref{eq:intermediate_Cauchy_Schwarz}, we state a helper lemma that bounds the distance between the subspaces spanned by the columns of $\bV_\pre$ and $\bhV_\pre$. 
Its proof is given in Appendix \ref{sec:consistency.2}. 
\begin{lemma} \label{lemma:subspace} 
    Consider the setup of Theorem \ref{thm:consistency}. Then for any $\zeta > 0$, we have w.p. at least $1 - 2\exp(-\zeta^2)$
    \begin{align}
        \| \Pc_{\hV_\epre} - \Pc_{V_\epre} \|_\eop
        \le \frac{C\sigma \left(\sqrt{T_0} + \sqrt{\numdonors} + \zeta \right)}{s_{r_\epre}},    
    \end{align} 
    where $s_{r_\epre}$ is the $r_\epre$-th singular value of $\Ex[\bYepred]$ with $r_\epre = \rank(\Ex[\bYepred])$, 
    and $C>0$ is an absolute constant.
\end{lemma} 
Applying Lemma \ref{lemma:subspace} with Assumption \ref{assumption:spectra}, we have 
\begin{align}\label{eq:wedin_balanced_spectra}
    \| \Pc_{\hV_\pre} - \Pc_{V_\pre} \|_\txtop 
    = O_p \left( \frac{\sqrt{r_\pre}}{\min\{\sqrt{T_0}, \sqrt{N_d}\}}\right). 
\end{align}
Substituting \eqref{eq:wedin_balanced_spectra} and the bound in Lemma \ref{lemma:param_est} into \eqref{eq:intermediate_Cauchy_Schwarz}, we obtain
\begin{align}
    (\Pc_{V_\pre} - \Pc_{\hV_\pre}) \Delta^{(i,d)}
    &= O_p\left(\frac{\sqrt{r_\pre \log(T_0 \numdonors)}}{\min\{\sqrt{T_0}, \sqrt{N_d}\}}  \left[ \frac{r^{3/4}_\pre }{T^{1/4}_0 \numdonors^{1/2}} + \frac{r^{3/2}_\pre }{\min\{T_0,\numdonors\}} \right] \right). \label{eq:dd.0} 
\end{align}

\medskip \noindent 
{\em Bounding term 2. } 
Since $\bhV_{\pre}$ is an isometry, it follows that 
\begin{align}\label{eq:thm1.2.1}
\| \Pc_{\hV_{\pre}} \Delta^{(i,d)}  \|_2^2 & 
= \|\bhV_{\pre}^\top \Delta^{(i,d)}  \|_2^2.
\end{align}
We upper bound $\|\bhV_{\pre}^\top \Delta^{(i,d)}  \|_2^2$ as follows: consider
\begin{align}
&\| \bYpredrpre \Delta^{(i,d)} \|_2^2 
= (\bhV_{\pre}^\top \Delta^{(i,d)})^\top \bhS_{r_\pre}^2 (\bhV_{\pre}^\top \Delta^{(i,d)})   
\geq \hs_{r_\pre}^2 \cdot \| \bhV_{\pre}^\top \Delta^{(i,d)} \|_2^2 \label{eq:key_singular_value_trick}. 
\end{align}
Using \eqref{eq:key_singular_value_trick} and \eqref{eq:thm1.2.1} together implies  
\begin{align}\label{eq:intermediate_bound_2_Lem10.1}
   \| \Pc_{\hV_{\pre}} \Delta^{(i,d)}  \|_2^2 \le \frac{\| \bYpredrpre \Delta^{(i,d)} \|_2^2}{\hs_{r_\pre}^2}. 
\end{align}
To bound the numerator in \eqref{eq:intermediate_bound_2_Lem10.1}, note 
\begin{align}
\| \bYpredrpre \Delta^{(i,d)} \|_2^2 
& \leq 2\| \bYpredrpre \hw^{(i,d)} -\Ex[Y_{\pre, i}] \|_2^2 
+ 2 \| \Ex[Y_{\pre, i}] -\bYpredrpre \tw^{(i,d)}\|_2^2
\\ & = 2 \| \bYpredrpre \hw^{(i,d)} -\Ex[Y_{\pre, i}] \|_2^2
+ 2 \| (\Ex[\bYpred] - \bYpredrpre) \tw^{(i,d)}\|_2^2,  \label{eq:intermediate_bound_1_Lem10.1}
\end{align}
where we have used \eqref{eq:pre_intervention_identification}.
To further upper bound the the second term on the RHS above, we use the following inequality: 
for any $\bA \in \Rb^{a \times b}$, $v \in \Rb^{b}$,
\begin{align}\label{eq:2.inf.ineq}
\| \bA v \|_2 & = \| \sum_{j=1}^b \bA_{\cdot j} v_j \|_2  \leq \left(\max_{j \le b} \| \bA_{\cdot j}\|_2 \right) \cdot \left(\sum_{j=1}^b |v_j|\right) 
= \| \bA\|_{2, \infty} \cdot \|v\|_1,
\end{align}
where $\|\bA\|_{2,\infty} = \max_{j} \| \bA_{\cdot j}\|_2$ and $\bA_{\cdot j}$ represents the $j$-th column of $\bA$. 
Thus,
\begin{align}\label{eq:2_infty_matrix}
\| (\Ex[\bYpred] - \bYpredrpre) \tw^{(i,d)}\|_2^2
\le 
\| \Ex[\bYpred] - \bYpredrpre \|_{2,\infty}^2 \cdot \|\tw^{(i,d)}\|_1^2. 
\end{align}
Substituting \eqref{eq:intermediate_bound_1_Lem10.1} into \eqref{eq:intermediate_bound_2_Lem10.1} and subsequently using \eqref{eq:2_infty_matrix} implies
\begin{align}\label{eq:thm1.2.2}
\| \Pc_{\hV_{\pre}} \Delta^{(i,d)}  \|_2^2  
&\leq \frac{2}{\hs_{r_\pre}^2} \Big( \| \bYpredrpre \hw^{(i,d)} -\Ex[Y_{\pre, i}] \|_2^2 
 + \| \Ex[\bYpred] - \bYpredrpre \|_{2,\infty}^2 \cdot \|\tw^{(i,d)}\|_1^2\Big).
\end{align}
Next, we bound $\| \bYpredrpre \hw^{(i,d)} -\Ex[Y_{\pre, i}] \|_2^2$.
To this end, observe that 
\begin{align}
\| \bYpredrpre \hw^{(i,d)} - Y_{\pre, i} \|_2^2 & = \| \bYpredrpre \hw^{(i,d)} - \Ex[Y_{\pre, i}] - \varepsilon_{\pre, i}\|_2^2 
\\& = \| \bYpredrpre \hw^{(i,d)} - \Ex[Y_{\pre, i}] \|_2^2 + \|\varepsilon_{\pre, i}\|_2^2 
 \\ &\quad - 2 \langle \bYpredrpre \hw^{(i,d)} - \Ex[Y_{\pre, i}], \varepsilon_{\pre, i} \rangle. 
\label{eq:thm1.2.3}
\end{align}
By the optimality conditions of $\hw^{(i,d)}$ and  \eqref{eq:pre_intervention_identification}, we have 
\begin{align}
\| \bYpredrpre \hw^{(i,d)} - Y_{\pre, i} \|_2^2 
&\leq \| \bYpredrpre \tw^{(i,d)} - Y_{\pre, i} \|_2^2
\\ &= \|\bYpredrpre \tw^{(i,d)}-  \Ex[Y_{\pre, i}] - \varepsilon_{\pre, i} \|_2^2  
\\ &= \|\bYpredrpre \tw^{(i,d)}-  \Ex[\bYpred] \tw^{(i,d)} - \varepsilon_{\pre, i} \|_2^2  
\\ &= \|( \bYpredrpre - \Ex[\bYpred])\tw^{(i,d)} \|_2^2  + \| \varepsilon_{\pre, i} \|_2^2 
 \\ &\quad - 2 \langle \bYpredrpre \tw^{(i,d)}-  \Ex[Y_{\pre, i}], \varepsilon_{\pre, i}\rangle.
\label{eq:thm1.2.4}
\end{align}
From \eqref{eq:thm1.2.3} and \eqref{eq:thm1.2.4}, we have 
\begin{align}
\| \bYpredrpre \hw^{(i,d)} - \Ex[Y_{\pre, i}] \|_2^2 
&\leq \| ( \bYpredrpre - \Ex[\bYpred])\tw^{(i,d)} \|_2^2 
 + 2 \langle \bYpredrpre \Delta^{(i,d)}, \varepsilon_{\pre, i} \rangle 
\\ &\leq \| \bYpredrpre - \Ex[\bYpred] \|_{2,\infty}^2 \cdot \|\tw^{(i,d)}\|_1^2 
 \\ &\quad + 2 \langle \bYpredrpre \Delta^{(i,d)}, \varepsilon_{\pre, i} \rangle,
\label{eq:thm1.2.5}
\end{align}
where we used \eqref{eq:2.inf.ineq} in the second inequality above.  
Using \eqref{eq:thm1.2.2} and \eqref{eq:thm1.2.5}, we obtain 
\begin{align} \label{eq:thm1.2}
 \| \Pc_{\hV_{\pre}} \Delta^{(i,d)}  \|_2^2 
& \leq \frac{4}{\hs_{r_\pre}^2} \left( \| \bYpredrpre - \Ex[\bYpred] \|_{2,\infty}^2 \cdot \|\tw^{(i,d)}\|_1^2
+\langle \bYpredrpre \Delta^{(i,d)}, \varepsilon_{\pre, i} \rangle\right).
\end{align}
We now state two helper lemmas that will help us conclude the proof. 
The proof of Lemmas \ref{lemma:hsvt} and \ref{lemma:annoying} are given in Appendices \ref{sec:proof_hsvt_lemma} and \ref{ssec:annoying}, respectively. 

\begin{lemma} \label{lemma:hsvt}
Let Assumptions \ref{assumption:sutva}, \ref{assumption:form}, \ref{assumption:conditional_mean_zero}, 
\ref{assumption:noise}, 
\ref{assumption:boundedness}, 
\ref{assumption:spectra} hold. 
Suppose $k = r_\epre$, where $k$ is defined as in \eqref{eq:pcr}. 
Then conditioned on $\cE$, we have 
$$
    \| \bYepredrpre - \Ex[\bYepred] \|_{2,\infty}
    = O_p \left( \frac{\sqrt{r_\epre T_0 \log(T_0 \numdonors)}}{\min\{\sqrt{T_0}, \sqrt{N_d}\}} \right). 
$$
\end{lemma}

\begin{lemma}  \label{lemma:annoying}
Let Assumptions \ref{assumption:sutva} to \ref{assumption:spectra} hold. 
Then, given $\bYepredrpre$ and conditioned on $\cE$, the following holds with respect to the randomness in $\varepsilon_{\epre, i}$: 
$$
     \langle \bYepredrpre \Delta^{(i,d)}, \varepsilon_{\epre, i} \rangle 
     = O_p \left( r_\epre + \sqrt{T_0} +  \| \bYepredrpre - \Ex[\bYepred] \|_{2,\infty} \cdot \| \tw^{(i,d)}\|_1 \right). 
$$
\end{lemma} 

Incorporating Lemmas \ref{lemma:singular_values}, \ref{lemma:coeff_bound}, \ref{lemma:hsvt}, \ref{lemma:annoying}, and Assumption \ref{assumption:spectra} 
into \eqref{eq:thm1.2}, we conclude
\begin{align}
&\Pc_{\hV_{\pre}} \Delta^{(i,d)}
\\ &= O_p \left(\frac{r_\pre \| \tw^{(i,d)}\|_1 \sqrt{\log(T_0 \numdonors)}}{\sqrt{N_d} \cdot \min\{\sqrt{T_0}, \sqrt{N_d} \}} + \frac{r_\pre}{\sqrt{N_d T_0}} + \frac{\sqrt{r_\pre}}{\sqrt{N_d} T_0^{\frac 1 4}} 
 + \frac{r^{3/4}_\pre \| \tw^{(i,d)}\|^{1/2}_1 \log^{1/4}(T_0 \numdonors)}{\sqrt{N_d} T^{1/4}_0 \min\{T^{1/4}_0, N^{1/4}_d \}}\right)
\\ &= O_p \left( \frac{\sqrt{r_\pre}}{\sqrt{N_d} T_0^{\frac 1 4}} + \frac{r^{3/2}_\pre  \sqrt{\log(T_0 \numdonors)}}{\sqrt{N_d} \cdot \min\{\sqrt{T_0}, \sqrt{N_d} \}}  \right). 
\label{eq:dd.1} 
\end{align} 

\vspace{5pt} 
\noindent 
{\em Collecting terms. } 
Combining \eqref{eq:dd.0} and \eqref{eq:dd.1}, we conclude
\begin{align}
&\Pc_{V_\pre} \Delta^{(i,d)} 
\\&= O_p \left( \frac{\sqrt{r_\pre}}{\sqrt{N_d} T_0^{\frac 1 4}} + \frac{r^{3/2}_\pre  \sqrt{\log(T_0 \numdonors)}}{\sqrt{N_d} \cdot \min\{\sqrt{T_0}, \sqrt{N_d} \}}   
+ \frac{\sqrt{r_\pre \log(T_0 \numdonors)}}{\min\{\sqrt{T_0}, \sqrt{N_d}\}}  \left[ \frac{r^{3/4}_\pre }{T^{1/4}_0 \numdonors^{1/2}} + \frac{r^{3/2}_\pre }{\min\{T_0,\numdonors\}} \right]
\right)
\\&= O_p \left( \frac{\sqrt{r_\pre}}{\sqrt{N_d} T_0^{\frac 1 4}} + \frac{r^{3/2}_\pre  \sqrt{\log(T_0 \numdonors)}}{\sqrt{N_d} \cdot \min\{\sqrt{T_0}, \sqrt{N_d} \}}   
+   \frac{r^{2}_\pre \sqrt{\log(T_0 \numdonors)}}{\min\{T_0^{3/2},\numdonors^{3/2}\}}
\right).
\end{align}
\end{proof} 

\subsection{Proof of Lemma \ref{lemma:subspace}} \label{sec:consistency.2} 
We recall the well-known singular subspace perturbation result by \cite{Wedin1972PerturbationBI}.  

\begin{lemma} [Wedin's Theorem \cite{Wedin1972PerturbationBI}] \label{thm:wedin}
Given $\bA, \bB \in \Rb^{m \times n}$, let $\bV, \bhV \in \Rb^{n \times n}$ denote their respective right singular vectors.
Further, let $\bV_{k} \in \Rb^{n \times k}$ (respectively, $\bhV_{k} \in \Rb^{n \times k}$) correspond to the truncation of $\bV$ (respectively, $\bhV$), respectively, that retains the columns corresponding to the top $k$ singular values of $\bA$ (respectively, $\bB$).  
Let $s_{i}$ represent the $i$-th singular value of $\bA$.
Then, 
$$
	 \| \Pc_{V_{k}} - \Pc_{\hV_{k}} \|_\eop  \leq \frac{2\|\bA - \bB\|_\eop} { s_{k} - s_{k+1}}.
$$
\end{lemma}

\begin{proof} 
Recall that $\bhV_\pre$ is formed by the top $r_\pre$ right singular vectors of $\bYpred$
and $\bV_\pre$ is formed by the top $r_\pre$ singular vectors of $\Ex[\bYpred]$. 
Therefore, Lemma \ref{thm:wedin} gives
$$
    \| \Pc_{\hV_\pre} - \Pc_{V_\pre} \|_\txtop \le 
    \frac{2 \| \bYpred - \Ex[\bYpred] \|_\txtop} {s_{r_\pre}},
$$
where we used $\rank(\Ex[\bYpred]) = r_\pre$ and hence $s_{r_{\pre}+1}=0$. 
By Assumption \ref{assumption:noise} and Lemma \ref{lemma:subg_matrix}, we can further bound the inequality above. 
In particular, for any $\zeta >0$, we have
w.p. at least $1-2\exp(-\zeta^2)$,
$$
    \| \Pc_{\hV_\pre} - \Pc_{V_\pre} \|_\txtop \le \frac{C \sigma (\sqrt{T_0} + \sqrt{\numdonors} + \zeta)}{s_{r_\pre}},
$$
where $C>0$ is an absolute constant. 
\end{proof} 

\subsection{Proof of Lemma \ref{lemma:hsvt}}\label{sec:proof_hsvt_lemma}
We first re-state Lemma 7.2 of \cite{pcr_aos}. 

\begin{lemma} [Lemma 7.2 of \cite{pcr_aos}]
Let the setup of Lemma \ref{lemma:hsvt} hold.
Then w.p. at least $1 - O(1 / (T_0 \numdonors)^{10})$,
\begin{align}
    &\|\Ex[\bYepred] - \bYepredrpre \|^2_{2, \infty} 
    \\& \quad \le C(\sigma) \left\{\frac{(T_0 + \numdonors) \cdot (T_0 + \sqrt{T_0}\log(T_0 \numdonors))}{s_{r_\epre}^2} + r_\epre + \sqrt{r_\epre }\log(T_0 \numdonors) \right\}
    + \frac{\log(T_0 \numdonors)}{\numdonors}, \qquad \label{eq:hsvt_heler_iew_pcr}
\end{align}
where $C(\sigma)$ is a constant that only depends on $\sigma$.
\end{lemma}

\begin{proof} 
We adapt the notation in \cite{pcr_aos} to that used in this paper.
In particular, $\bX = \Ex[\bYpred]$, $\widetilde{\bZ}^{r} = \bYpredrpre $, where 
$\bX, \widetilde{\bZ}^{r}$ are the notations used in \cite{pcr_aos} with $r = r_\pre$.

Next, we simplify \eqref{eq:hsvt_heler_iew_pcr} using Assumption \ref{assumption:spectra}. 
W.p. at least $1 - O(1 / (T_0 \numdonors)^{10})$, we have
\begin{align}\label{eq:hsvt_heler_iew_pcr_2}
    &\|\Ex[\bYpred] - \bYpredrpre \|^2_{2, \infty} \le C(\sigma) \left(\frac{r_\pre T_0\log(T_0 \numdonors)}{\min\{T_0, \numdonors\}}\right).
\end{align}
This concludes the proof of Lemma \ref{lemma:hsvt}.
\end{proof}

\subsection{Proof of Lemma \ref{lemma:annoying}}\label{ssec:annoying}
Throughout this proof, $C, c >0$ will denote absolute constants, which can change from line to line or even within a line.
\medskip 
\begin{proof} 
Recall $\hw^{(i,d)} = \bhV_\pre \bhS^{-1}_\pre \bhU^\top_\pre Y_{\pre, i}$ and $Y_{\pre, i} = \Ex[Y_{\pre, i}] + \varepsilon_{\pre, i}$. 
Thus, 
\begin{align} \label{eq:rewrite} 
\bYpredrpre \hw^{(i,d)} = \bhU_\pre\bhS_\pre \bhV_\pre^\top \bhV_\pre \bhS_\pre^{-1} \bhU_\pre^\top Y_{\pre, i}  
=  \Pc_{\hU_\pre} \Ex[Y_{\pre, i}] +  \Pc_{\hU_\pre} \varepsilon_{\pre, i}.
\end{align}
We then obtain 
\begin{align}
\langle \bYpredrpre (\hw^{(i,d)} - \tw^{(i,d)}),  \varepsilon_{\pre, i} \rangle 
 &= \langle  \Pc_{\hU_\pre} \Ex[Y_{\pre, i}], \varepsilon_{\pre, i} \rangle 
 + \langle   \Pc_{\hU_\pre} \varepsilon_{\pre, i}, \varepsilon_{\pre, i} \rangle 
  \\ &\quad - \langle \bhU_\pre\bhS_\pre \bhV_\pre^\top \tw^{(i,d)}, \varepsilon_{\pre, i} \rangle.  \label{eq:lem4.0}
\end{align}
Note that $\varepsilon_{\pre, i}$ is independent of $\bYpred$, and thus also independent of $\bhU_\pre, \bYpredrpre$. 
Therefore, 
\begin{align}
    \Ex\Big[\langle  \Pc_{\hU_\pre} \Ex[Y_{\pre, i}], \varepsilon_{\pre, i} \Big] &= 0, \label{eq:lemma10.5_exp_2}
    \\ \Ex\Big[\langle \bhU_\pre\bhS_\pre \bhV_\pre^\top \tw^{(i,d)}, \varepsilon_{\pre, i}\rangle \Big] &= 0 \label{eq:lemma10.5_exp_3}
\end{align}
Moreover, using the cyclic property of the trace operator, we obtain 
\begin{align}
\Ex\big[\langle \Pc_{\hU_\pre} \varepsilon_{\pre, i}, \varepsilon_{\pre, i}  \rangle\big] 
&= \Ex\big[ \varepsilon^\top_{\pre, i}  \Pc_{\hU_\pre} \varepsilon_{\pre, i} \big] 
= \Ex \big[\tr(\varepsilon^\top_{\pre, i}  \Pc_{\hU_\pre} \varepsilon_{\pre, i})\big] 
%
%
\\ &= \tr(\Ex\big[ \varepsilon_{\pre, i} \varepsilon^\top_{\pre, i}\big] \Pc_{\hU_\pre}) 
= \tr(\sigma^2 \Pc_{\hU_\pre})  
\\ &=  \sigma^2 \cdot r_\pre. 
\label{eq:thm1.lem3.0}
\end{align}
Accordingly, we have 
\begin{align}\label{eq:thm1.lem3.1}
\Ex \left[\langle \bYpredrpre \Delta^{(i,d)},  \varepsilon_{\pre, i} \rangle \right] & = \sigma^2 \cdot r_\pre. 
\end{align}
Next, to obtain high probability bounds for the inner product term, we use the following lemmas, the proofs of which can be found in \cite[Appendix D]{pcr_aos}. 

\begin{lemma} [Modified Hoeffding's Lemma] \label{lemma:hoeffding_random} 
Let $X \in \Rb^n$ be r.v. with independent mean-zero sub-Gaussian random coordinates with $\| X_i \|_{\psi_2} \le K$.
Let $a \in \Rb^n$ be another random vector that satisfies $\|a\|_2 \le b$ for some constant $b \ge 0$.
Then for any $\zeta \ge 0$, 
$
	\Pb \Big( \Big| \sum_{i=1}^n a_i X_i\Big| \ge \zeta \Big) \le 2 \exp\Big(-\frac{c\zeta^2}{K^2 b^2} \Big),
$
where $c > 0$ is a universal constant. 
\end{lemma}

\begin{lemma} [Modified Hanson-Wright Inequality] \label{lemma:hansonwright_random}
Let $X \in \Rb^n$ be a r.v. with independent mean-zero sub-Gaussian coordinates with $\|X_i\|_{\psi_2} \le K$. 
Let $\bA \in \Rb^{n \times n}$ be a random matrix satisfying $\|\bA\|_{\eop}  \le a$ and $\|\bA\|_F^2 \, \le b$ for some $a, b \ge 0$.
Then for any $\zeta \ge 0$,
\begin{align*}
	\Pb \left( \abs{ X^\top \bA X - \Ex[X^\top \bA X] } \ge \zeta \right) &\le 2 \exp \left\{ -c \min\left(\frac{\zeta^2}{K^4 b}, \frac{\zeta}{K^2 a} \right) \right\}. 
\end{align*}
\end{lemma}

Using Lemma \ref{lemma:hoeffding_random} and \eqref{eq:lemma10.5_exp_2}, and 
Assumptions \ref{assumption:form} and \ref{assumption:noise}, 
it follows that for any $\zeta > 0$
\begin{align}\label{eq:lem4.1}
\Pb\left(  \langle \Pc_{\hU_\pre} \Ex[Y_{\pre, i}], \varepsilon_{\pre, i} \rangle  \geq \zeta \right) & \leq  \exp\left( - \frac{c \zeta^2}{T_0 \sigma^2 } \right).
\end{align}
Note that the above also uses
$
    \| \Pc_{\hU_\pre} \Ex[Y_{\pre, i}] \|_2
    \le \| \Ex[Y_{\pre, i}] \|_2 
    \le \sqrt{T_0}, 
$
which follows from the fact that $\|\Pc_{\hU_\pre}\|_\txtop \le 1$ and Assumption \ref{assumption:boundedness}.
Further, \eqref{eq:lem4.1} implies  
\begin{align}\label{eq:lem4.1_Op}
    \langle \Pc_{\hU_\pre} \Ex[Y_{\pre, i}], \varepsilon_{\pre, i} \rangle  = O_p \left( \sqrt{T_0} \right). 
\end{align}
Similarly, using \eqref{eq:lemma10.5_exp_3}, we have for any $\zeta > 0$
\begin{align}\label{eq:lem4.2}
&\Pb\left(  \langle \bhU_\pre\bhS_\pre \bhV_\pre^\top\tw^{(i,d)}, \varepsilon_{\pre, i} \rangle  \geq \zeta \right)
\\ &\quad
 \leq  \exp\left\{ - \frac{c \zeta^2}{\sigma^2 (T_0 +\| \bYpredrpre - \Ex[ \bYpred]\|^2_{2,\infty} \cdot \|\tw^{(i,d)}\|_1^2)  }  \right\},
\end{align}
where we use the fact that
\begin{align}
\| \bhU_\pre\bhS_\pre \bhV_\pre^\top\tw^{(i,d)}\|_2 
&= \| \bYpredrpre \tw^{(i,d)} \pm  \Ex[Y_{\pre, i}] \|_2 
\\ & = \| (\bYpredrpre - \Ex[\bYpred])\tw^{(i,d)} + \Ex[Y_{\pre, i}]\|_2 
\\ &\leq \| (\bYpredrpre - \Ex[\bYpred]) \tw^{(i,d)}\|_2 + \|\Ex[Y_{\pre, i}] \|_2 \nonumber \\
& \leq \|\bYpredrpre - \Ex[\bYpred]\|_{2,\infty} \cdot \| \tw^{(i,d)}\|_1 + \sqrt{T_0}.
\end{align}
In the inequalities above, 
we use $\Ex[\bYpred] \cdot \tw^{(i,d)} = \Ex[\bYpred] \cdot w^{(i,d)} = \Ex[Y_{\pre, i}]$, which follows from \eqref{eq:theta.0} and \eqref{eq:2.inf.ineq}. 
Then, \eqref{eq:lem4.2} implies 
\begin{align}\label{eq:lem4.2_Op}
    \langle \bhU_\pre\bhS_\pre \bhV_\pre^\top\tw^{(i,d)}, \varepsilon_{\pre, i} \rangle  
    = O_p\left( 
    \|\bYpredrpre - \Ex[\bYpred]\|_{2,\infty} \cdot \| \tw^{(i,d)}\|_1 + \sqrt{T_0}
\right)
\end{align}
Finally, using Lemma \ref{lemma:hansonwright_random}, \eqref{eq:thm1.lem3.1},  Assumptions \ref{assumption:form} and \ref{assumption:noise}, we have for any $\zeta > 0$
\begin{align}\label{eq:lem4.3}
\Pb\left(  \langle  \Pc_{\hU_\pre} \varepsilon_{\pre, i}, \varepsilon_{\pre, i} \rangle   \geq \sigma^2 r_\pre + \zeta \right) & \leq  \exp\left\{ - c \min\left(\frac{\zeta^2}{\sigma^4 r_\pre}, \frac{\zeta}{\sigma^2}\right)\right\},
\end{align}
where we have used $\|\Pc_{\hU_\pre}\|_\txtop \le 1$ and $\|\Pc_{\hU_\pre}\|^2_F = r_\pre$.
Then, \eqref{eq:lem4.3} implies 
\begin{align}\label{eq:lem4.3_Op}
    \langle \Pc_{\hU_\pre} \varepsilon_{\pre, i}, \varepsilon_{\pre, i}\rangle  = O_p\left( r_\pre \right).
\end{align}
From \eqref{eq:lem4.0}, \eqref{eq:lem4.1_Op}, \eqref{eq:lem4.2_Op}, and \eqref{eq:lem4.3_Op}, we arrive at our desired result.  
%
%
\end{proof}

%% file: content/nonlinear_lvm.tex
\section{Smooth Non-linear Latent Variable Models} \label{sec:nonlinear_lvm}
We expound on Remark~\ref{remark:nonlinear_lvms}. 
Specifically, we show that sufficiently continuous non-linear latent variable models (LVMs) are arbitrarily well-approximated by linear factor models of appropriate dimension as $N, T$ grow.
In this sense, the factor model in Assumption \ref{assumption:form} serves as a good approximation for non-linear LVMs when we have many units and time periods.

\subsection{Definitions}
We first precisely define the notion of sufficiently continuous.
\begin{definition}[H\"older continuity]\label{def:holder}
{\em
The H\"older class $\Hc(q,S,C_H)$ on $[0, 1)^q$ is the set of functions $h: [0, 1)^q \to \mathbb{R}$ whose partial derivatives satisfy
$$
\sum_{s: |s| = \lfloor S \rfloor} \frac{1}{s!} |\nabla_s h(\mu) - \nabla_s h(\mu') | \le C_H \norm{\mu - \mu'}_{\max}^{S - \lfloor S \rfloor},\quad \forall \mu, \mu' \in [0, 1)^q,
$$
where $\lfloor S \rfloor$ denotes the largest integer strictly smaller than $S$.
Here, $s = (s_1, \dots, s_q)$ is a multi-index with $s_\ell \in \Nb$, $|s| = \sum^q_{\ell = 1} s_\ell$,
and $\nabla_s h(\cdot) = \frac{\partial^{s_1} \dots \partial^{s_q}}{\partial \mu_{1} \dots \partial \mu_{q}} h(\cdot)$. 
}
\end{definition}

Next, we define ``smooth'' non-linear LVMs as follows.
\begin{definition}[``Smooth'' non-linear LVMs]\label{def:smooth_LVM}
{\em 
The potential outcomes are said to satisfy a smooth non-linear factor model if 
(i) $\Ex[Y_{ti}^{(d)}] = f(\rho^{(d)}_{t}, \phi_i)$, where $f: [0, 1)^{2q} \to \Rb$, and $\phi_i, \rho^{(d)}_{t} \in [0, 1)^q$;
(ii) $f(\rho^{(d)}_{t}, \cdot) \in \Hc(q,S,C_H)$ for all $t, d \in [T] \times [D]_0$.
}
\end{definition}

Intuitively, Definition~\ref{def:smooth_LVM} states that for a pair of units $(i_1, i_2)$, we have $\Ex[Y_{ti_1}^{(d)}] \approx \Ex[Y_{ti_2}^{(d)}]$ for all $(t,d)$ if $\phi_{i1} \approx \phi_{i2}$. 
The dimension $q$ of the latent variables $(\phi_i, \rho^{(d)}_{t})$ encodes the complexity of representing the various characteristics of units, time, and treatments that affect the potential outcomes.
Note that a linear factor model is a special case of a non-linear LVM where $f(\rho^{(d)}_{t}, \phi_i) = \langle \rho^{(d)}_{t}, \phi_i \rangle$.
Such a model satisfies Definition \ref{def:holder} for all $S \in \mathbb{N}$ and $C_H = C$, for some absolute finite, positive constant $C$. 

\subsection{Formal Result} 
Proposition \ref{prop:factor_model_approx} establishes that linear factor models of appropriate dimension also provide a good approximation for non-linear LVMs. 
\begin{proposition}[\cite{pcr_aos, xu2017rates}]\label{prop:factor_model_approx}
Suppose the potential outcomes satisfy Definition \ref{def:holder} for some fixed $\mathcal{H}(q,S,C_H)$.
Then for any $\delta>0$, there exists latent variables $u^{(d)}_t, v_i \in \Rb^r$ such that for all $(t,i,d)$, we have $| \Ex[Y_{ti}^{(d)}] - \sum^r_{\ell = 1} u^{(d)}_{t\ell} v_{i\ell} | \leq \Delta_E$, where $r \le C \cdot \delta^{-q}$ and $\Delta_E \le C_H \cdot \delta^S$; here, $C$ is allowed to depend on $(q,S)$. 
\end{proposition}

By choosing $\delta = \min\{N, T\}^{-c / q}$ for $c \in (0, 1)$,  
Proposition \ref{prop:factor_model_approx} implies that $\Delta_E \to 0$ as $N, T$ grow and $r \ll \min\{N, T\}$.
From this, we identify that linear latent factor models of appropriate dimension serve as a good approximation for non-linear LVMs as $N, T$ grow.

Further, we note that H\"older continuous functions are closed under composition.
For example, if $\Ex[Y^{(d)}_{ti}]$ satisfies Definition \ref{def:holder}, then $\log(\Ex[Y^{(d)}_{ti}])$ also satisfies it as $\log(\cdot)$ is an analytic function.
In contrast, a two-way fixed effects model, which is the canonical model for methods such as Difference-in-Differences, 
is of the form 
\begin{align}\label{eq:two_way}
    \Ex[Y^{(d)}_{ti}] = u^{(d)}_t + v_i.
\end{align}
This is not robust to the choice of parameterization, e.g., if \eqref{eq:two_way} holds for $\Ex[Y^{(d)}_{ti}]$ then 
it is not guaranteed to hold for $\log(\Ex[Y^{(d)}_{ti}])$.
%

%% file: content/covariates.tex
\section{Incorporating Covariates} \label{sec:covariates}
In this section, we discuss how meaningful unit-level covariates can help improve estimation. 
Let $\bX = [X_{ki}] \in \Rb^{K \times N}$ denote the matrix of covariates across units, i.e.,
$X_{ki}$ denotes the $k$-th descriptor (or feature) of unit $i$. 

\subsection{Assumptions}
One approach towards incorporating covariates into the \SI~estimator
described in Section \ref{sec:algo}, is to impose the following structure on $\bX$. 

\begin{assumption} [Covariate structure] \label{assump:covariates}
For any $k \in [K]$ and $i \in [N]$, let $X_{ki} = \langle \varphi_k, v_i \rangle + \zeta_{ki}$. 
Here, $\varphi_k \in \Rb^r$ represents a latent factor specific to descriptor $k$, 
$v_i \in \Rb^r$ is the unit latent factor as defined in \eqref{eq:form.0}, 
and $\zeta_{ki} \in \Rb$ is a mean zero measurement noise specific 
to descriptor $k$ and unit $i$. 
\end{assumption} 

The key structure of Assumption~\ref{assump:covariates} is that the covariates $X_{ki}$ share the same 
latent unit factors, $v_i$, as the potential outcomes $Y^{(d)}_{tn}$.
Thus, given a target unit $i$ and subgroup $\Ic^{(d)}$, this allows us to express unit $i$'s covariates as a linear 
combination of the covariates associated with  units within $\Ic^{(d)}$ via the {\em same} linear model that describes 
the relationship between their respective potential outcomes (formalized in Proposition~\ref{prop:covariates} below). 
One notable flexibility of our covariate model is that the observations of covariates can be {\em noisy} due to the presence of measurement
noise $\zeta_{ki}$.

\begin{proposition} \label{prop:covariates}
Given $(i,d)$, let Assumptions \ref{assumption:linear} and  \ref{assump:covariates} hold. 
Then, conditioned on $\Ec_\varphi = \Ec \cup \{\varphi_k: k \in [K]\}$, we have for all $k$, 
\begin{align}
    \Ex[X_{ki} \,| \, \Ec_\varphi] = \sum_{j \in \Ic^{(d)}} w_j^{(i,d)} \cdot \Ex[X_{kj} \,|\, \Ec_\varphi], 
\end{align}
where $w^{(i,d)}$ defined as in Assumption~\ref{assumption:linear}.
\end{proposition}  

\begin{proof} 
Plugging Assumption \ref{assumption:linear} into Assumption \ref{assump:covariates} completes the proof.
\end{proof}

\subsection{Description of Estimator with Covariates} 
Proposition \ref{prop:covariates} suggests a natural modification to step (i) of the \SI~estimator. 
Similar to the estimation procedure of \cite{abadie2}, we propose to append the covariates to the pre-treatment outcomes. 
Formally, let $X_i = [X_{ki}] \in \Rb^{K}$ denote the vector of covariates associated with the unit $i$; 
analogously, let  $\bX_{\Ic^{(d)}} = [X_{kj}: j \in \Ic^{(d)}] \in \Rb^{K \times N_d}$
denote the matrix of covariates for the units in $\Ic^{(d)}$. 
Define $Z_i = [Y_{\pre, i}^\top, X_i^\top]^\top \in \Rb^{T_0 + K}$ 
and 
$\bZ_{\Ic^{(d)}} = [\bY^\top_{\pre, \Ic^{(d)}}, \bX^\top_{\Ic^{(d)}}]^\top \in \Rb^{(T_0 + K) \times N_d}$ 
as the concatenation of pre-treatment outcomes and covariates for unit $i$ and subgroup $\Ic^{(d)}$, respectively. 
The \SI~estimator then proceeds by learning $\hw^{(i,d)}$ from $(Z_i, \bZ_{\Ic^{(d)}})$ and continues with step (ii) as described in Section~\ref{ssec:method}.


\subsection{Theoretical Implications} 
Consider the \SI-\PCR~estimator with covariates. 
Then, it can be verified that Theorem~\ref{thm:consistency} continues to hold with $T_0$ replaced with $T_0 + K$. 
Therefore, adding covariates into the model-learning stage augments the data size, which can improve the estimation rates.


%% file: content/causal_inference_tensor_completion.tex
\section{Causal Inference and Tensor Completion}\label{sec:tensor_framework} 
We elaborate on our connection between causal inference and tensor completion. 
Recall Assumption~\ref{assumption:sutva}, which connects our observed tensor, $\bY$, with the underlying potential outcomes tensor, $\bY^*$. 
As elucidated in Section~\ref{sec:ci.tc}, the fundamental task in causal inference of estimating unobserved potential outcomes can be achieved through tensor completion. 
Thus, as suggested by Figure \ref{fig:tensor.obs}, different observational and experimental studies can be equivalently posed as tensor completion under different sparsity patterns. 


\subsection{Tensor Factor Models and Algorithm Design}\label{sec:low_rank_tensor_factor_models}
Recall that the factorization given by \eqref{eq:form.0} in Assumption \ref{assumption:form} is implied by the factorization assumed by a low-rank tensor given by \eqref{eq:tensor.fm}.
In particular, Assumption \ref{assumption:form} does not require the additional factorization of the time-treatment factor $u^{(d)}_t$ as $ u_t \circ \lambda_d$, where $\circ$ denotes the entry-wise Hadamard product. 
Hence, we ask whether it is feasible to design estimators that exploit an implicit factorization of $u^{(d)}_t =  u_t \circ \lambda_d $.

Indeed, a recent follow-up work of \cite{squires2021causal} finds that a variant of \SI~that regresses across treatments as opposed to units leads to better empirical results in the setting of single cell therapies. 
At the same time, \cite{christina_tensor} directly exploits the tensor factor structure in  \eqref{eq:tensor.fm} to provide max-norm error bounds for tensor completion under a missing completely at random (MCAR) setting. 
We leave as an open question whether one can exploit the latent factor structure in  \eqref{eq:tensor.fm} over that in Assumption \ref{assumption:form} to operate under less stringent causal assumptions or data requirements and/or identify and estimate different causal estimands.

\subsection{A New Tensor Completion Toolkit for Causal Inference}\label{sec:new_tensor_toolkit}
The tensor completion problem provides an expressive formal framework for a large number of emerging applications. 
In particular, this literature quantifies the number of samples required and the computational complexity of different estimators to achieve theoretical guarantees for a given error metric (e.g., Frobenius-norm error over the entire tensor). 
Indeed, this trade-off is of central importance in the emerging sub-discipline at the interface of computer science, statistics, and machine learning. 
Given our preceding discussions connecting causal inference with tensor completion, a natural question to ask is whether we can apply the current toolkit of tensor completion methods to understand statistical and computational trade-offs in causal inference.  

We believe a direct transfer of tensor completion techniques and analyses to causal inference problems is not immediately possible in general for a couple of reasons. 
First, most tensor completion results assume MCAR entries. 
As seen in Figure \ref{fig:tensor.obs}, 
causal inference settings often induce missing not at random (MNAR) sparsity patterns. 
Second, typical tensor completion results hold for the Frobenius-norm error across the entire tensor. 
By contrast, a variety of meaningful causal estimands require more refined error metrics over subsets of the tensor.

Hence, we pose two important and related open questions:  
(i) what block sparsity patterns and structural assumptions on the potential outcomes tensor allow for faithful recovery with respect to a meaningful error metric for causal inference, 
and (ii) if recovery is possible, what are the fundamental statistical and computational trade-offs that are achievable? 
An answer to these questions can aid in bridging causal inference with tensor completion, as well as computer science and machine learning more broadly.

%% file: content/proofs_normality.tex
\section{Inference} \label{sec:proof.normality}

\subsection{Asymptotic Normality} 
We show how asymptotic normality can be established under our operating assumptions. 

\begin{theorem} \label{thm:normality}
Given $(i,d)$, let the setup of Theorem \ref{thm:consistency} hold. Let $ \tw^{(i,d)}$ be defined as in \eqref{eq:wtilde} and suppose we condition on $\cE$. 
Let the following conditions hold:
(i) $T_0,  \numdonors, T_1 \to \infty$,
(ii) 
\begin{align} \label{eq:cond.ii}
    \frac{ r^{3/2}_\epre \sqrt{\log(T_0 \numdonors)}}{\| \tw^{(i,d)}\|_2 \cdot \min\{T_0,\numdonors, T^{1/4}_0 \numdonors^{1/2}\} }  = o(1),
\end{align}
and (iii)
\begin{align} \label{eq:cond.iii}
\frac{1}{\sqrt{T_1} \|\tw^{(i,d)}\|_2} \sum_{t > T_0}  \left \langle \Ex[Y_{t, \Ic^{(d)}}],  \bV_\epre \bV_\epre^\top \left(\hw^{(i,d)} - \tw^{(i,d)} \right) \right \rangle &= o_p(1),
\end{align}
where $\hw^{(i,d)}$ is given by \eqref{eq:si.linear_model}. 
Then, we have  
\begin{align}\label{eq:gaussian.limit}
\frac{\sqrt{T_1}}{\sigma \|\tw^{(i,d)}\|_2}  \Big(\htheta_i^{(d)} - \theta_i^{(d)}\Big) 
&~\xrightarrow{d}~ \mathcal{N} \left(0, 1\right). 
\end{align}
Here, $O_p(\cdot)$ is defined with respect to the sequence $\min\{N_d, T_0\}$.
\end{theorem}

A few remarks are in order. 
First, as opposed to Theorem~\ref{thm:consistency} that establishes consistency under a fixed $T_1$, Theorem~\ref{thm:normality} proves asymptotic normality with a growing $T_1$. 
Second, we note that \eqref{eq:cond.ii} can be equivalently stated as 
$$ \| \tw^{(i,d)}\|_2 \gg \left((r^{3/2}_\pre \sqrt{\log(T_0 \numdonors)}) / \min\{T_0,\numdonors, T^{1/4}_0 \numdonors^{1/2}\} \right),$$
i.e., we require the $\ell_2$-norm of $\tw^{(i,d)}$ to be sufficiently large.
Next, for \eqref{eq:cond.iii} to hold, we require that the estimation error between $\hw^{(i,d)}$ and $\tw^{(i,d)}$, projected onto the rowspace of $\Ex[\bYpred | \cE]$, to be vanishing sufficiently quickly. 
Finally, by \eqref{eq:gaussian.limit}, we see that any formulation of the \SI~estimator in which \eqref{eq:cond.ii} and \eqref{eq:cond.iii} simultaneously hold leads to an estimate that is asymptotically normal around the causal estimand. 
In light of this, we describe one such formulation. 

\subsection{Modified \SI-\PCR~Estimator} \label{sec:mod.pcr}
We present a slight modification to the \SI-\PCR~estimator and show that this strategy obeys our desired condition in \eqref{eq:cond.iii}.

\subsubsection{Description of Estimator}  \label{sec:mod.pcr.algo}
Recall the notation set in Section~\ref{sec:algo}. 
At a high-level, the modified estimator proceeds in the exact same manner as described in Section~\ref{ssec:method} with one change to \eqref{eq:pcr} of step (i). 

Let $\bYpred^k = \sum_{\ell=1}^k \hs_\ell \hu_\ell \hv_\ell^\top$ denote the rank-$k$-approximation of $\bYpred$.
We then construct a sub-sampled matrix that retains the columns of $\bYpred^k$ within an index set $\Omega \subset \cI^{(d)}$. 
We denote the resulting matrix as $\bYpred^{(k,\Omega)} \in \Rb^{T_0 \times |\Omega|}$. 
Next, we replace \eqref{eq:pcr} with 
\begin{align}
	\hw^{(i,d,\Omega)} = \left( \bYpred^{(k, \Omega)} \right)^\dagger \Yipre \in \Rb^{|\Omega|}. \label{eq:pcr.variant}
\end{align}
%
%
The estimation strategy continues as per step (ii) of Section~\ref{ssec:method} with $\hw^{(i,d,\Omega)}$ and $\bYpostd^\Omega \in \Rb^{T_1 \times |\Omega|}$, which is defined as the sub-matrix of $\bYpostd$ that retains the columns within $\Omega$. 

\subsubsection{Theoretical Implications} 
Let $\bV_{\cI^{(d)}} = [v_j: j \in \cI^{(d)}] \in \Rb^{\numdonors \times r}$.
We define $\bV_{\cI^{(d)}}^\Omega \in \Rb^{|\Omega| \times r}$ as the sub-sampled matrix that retains the rows of $\bV_{\cI^{(d)}}$ within $\Omega$. 
To adapt the results in Section~\ref{sec:formal} to our new setting, note that if $\Omega$ is chosen such that $\text{span}(\{v_j: j \in \Omega\})$ is an $r_\pre$-dimensional subspace of $\Rb^{\numdonors}$ (i.e., $\rank(\bV_{\cI^{(d)}}^\Omega) = r_\pre$), then there exists a $w^{(i,d,\Omega)} \in \Rb^{|\Omega|}$ such that 
\begin{align}
	v_i = \bV_{\cI^{(d)}}^\top \cdot w_j^{(i,d)} = \left(\bV_{\cI^{(d)}}^\Omega \right)^\top \cdot w_j^{(i,d,\Omega)}, \label{eq:image.compressed}
\end{align} 
where $w^{(i,d)} \in \Rb^{\numdonors}$ is defined as in Assumption~\ref{assumption:linear} (and Theorem~\ref{thm:identification}). 

Intuitively, our modified strategy anchors on the concept that only $r_\pre$ donors are needed to recreate our target unit $i$ since $\bV_{\cI^{(d)}}$ has rank $r_\pre$; hence, it suffices to use a subset of donors defined over $\Omega$ rather than utilizing all $N_d$ donor units, provided $\Omega$ is chosen such that \eqref{eq:image.compressed} holds. 
With this intuition, we are now ready to present the main result of this section. 
Below, let $\bMpred = \Ex[\bYpred | \mathcal{E}]$ and define $\bMpred^{\Omega} \in \Rb^{T_0 \times |\Omega|}$ as the sub-sampled matrix that retains the columns of $\bMpred$ within $\Omega$. 

\begin{lemma} \label{lemma:cond.iii.compressed}
Given $(i,d)$, let Assumptions~\ref{assumption:sutva} to \ref{assumption:subspace} hold. 
Consider the modified \SI-\PCR~estimator with $\hw^{(i,d,\Omega)}$ and $w^{(i,d,\Omega)}$ as defined as in \eqref{eq:pcr.variant} and \eqref{eq:image.compressed}, respectively.
Suppose $k$ and $\Omega$ are chosen such that for appropriate constant $C$, we have
\begin{itemize}
	\item[(i)] $k = r_\epre = \rank(\bMepred)$, 
	\item[(ii)] $|\Omega| = r_\epre$ with $\rank(\bYepred^{(r_\epre, \Omega)}) = r_\epre$ and $\rank(\bMepred^\Omega) = r_\epre$,
	\item[(iii)] Assumption~\ref{assumption:spectra} holds with respect to $\bMepred^\Omega$ and $s^\Omega_{r_\epre} \ge C \sigma \sqrt{T_0}$,
	\item[(iv)] $N_d < T_0$,
        \item[(v)] $\| w^{(i,d,\Omega)}) \|_2 \ge C r^{-1/2}_\epre$,
        \item[(vi)] $r^2_\epre = o\left( \frac{\min\{T_0,\numdonors, T^{1/4}_0 \numdonors^{1/2}\}}{\sqrt{\log(T_0 \numdonors)}} \right)$,
        \item[(vii)] $T_1 = o\left( \min\left\{ \frac{N_d}{\log(T_0 N_d) r^4_\epre}, \frac{T^{1/2}_0}{r^2_\epre}  \right\} \right)$. 
\end{itemize}
Then conditioned on $\cE$ and $\Omega$, we have  
\begin{align} \label{eq:cond.ii_PCR_variant}
    \frac{ r^{3/2}_\epre \sqrt{\log(T_0 \numdonors)}}{\| w^{(i,d,\Omega)}\|_2 \cdot \min\{T_0,\numdonors, T^{1/4}_0 \numdonors^{1/2}\} }  = o(1),
\end{align}
and 
\begin{align} \label{eq:cond.iii_PCR_variant}
\frac{1}{\sqrt{T_1} \|w^{(i,d,\Omega)}\|_2} \sum_{t > T_0}  \left \langle \Ex[Y^\Omega_{t, \Ic^{(d)}}],  \left( \bM^\Omega_{\epre, \Ic^{(d)}} \right)^\dagger  \bM^\Omega_{\epre, \Ic^{(d)}}\left(\hw^{(i,d,\Omega)} - w^{(i,d,\Omega)} \right) \right \rangle &= o_p(1), 
\end{align}
where $\Ex[Y^\Omega_{t, \Ic^{(d)}}] = [\Ex[Y_{tjd}]: j \in \Omega] \in \Rb^{r_\epre}$ and $O_p(\cdot)$ is defined with respect to $\numdonors$. 
\end{lemma} 
%

Lemma~\ref{lemma:cond.iii.compressed} states that the modified \SI-\PCR~estimator satisfies conditions \eqref{eq:cond.ii} and \eqref{eq:cond.iii} of Theorem \ref{thm:normality} to achieve asymptotic normality. 
This enables us to conduct valid inference via the following confidence interval: 
for $\gamma \in (0,1)$, 
\begin{align} \label{eq:finite_sample_CI}
	\mathcal{CI}(\gamma) = \left[ \htheta_i^{(d, \Omega)} \pm \frac{z_{\gamma/2} \cdot \widehat{\sigma}^\Omega  \cdot \|\hw^{(n, d, \Omega)}\|_2}{\sqrt{T_1}} \right], 
\end{align}
where $z_{\gamma/2}$ is the upper $\gamma/2$ quantile of the standard Normal and 
\begin{align}
	\hsigma^\Omega = \frac{1}{T_0} \| \Yipre - \bYpred^{(r, \Omega)} \cdot \hw^{(i,d,\Omega)} \|_2^2.
\end{align}  
We study properties of the proposed confidence interval below. 

\subsubsection{Simulation Study} 
This section presents a simulation to study the coverage properties of the confidence interval in \eqref{eq:finite_sample_CI} based on the modified \SI-\PCR~estimator. 

\paragraph{Data generating process.} 
We vary $T_0 \in \{200, 400, 600, 800, 1000\}$ and choose $r = 5$. 
Motivated by our conditions in Lemma~\ref{lemma:cond.iii.compressed}, we choose $N_d = T_0/2$ and $T_1 = \sqrt{T_0}$ for each $T_0$. 

\bigskip                                                                                                                                                                                                                                                                  
\noindent {\em Latent factors.} 
As in Section~\ref{sec:sim.dgp}, we generate the latent factors associated with the donor units, $\bV_{\cI^{(d)}} \in \Rb^{N_d \times r}$, by independently sampling its entries from a standard Normal. 
Then, we sample $w^{(i,d)} \in \Rb^\numdonors$ by drawing its entries from a uniform distribution over $[0,1]$ and normalizing it to have unit norm. 
We define the target unit latent factor as $v_i = \bV_{\cI^{(d)}}^\top \cdot w^{(i,d)}$. 

We define the latent factors for  the pre-treatment outcomes under control as $\bU^{(0)}_\pre$, where its entries are i.i.d. samples from a standard Normal. 
Next, we sample $\bPhi \in \Rb^{T_1 \times r}$ whose entries are i.i.d. draws from a uniform distribution over $[0, 1]$. 
We construct the post-treatment latent factors of treatment $d$ as $\bU_\post^{(d)} = \bPhi \cdot (\bU_\pre^{(0)})^\dagger \bU^{(0)}_\pre$. 
We define our causal estimand as $\theta_i^{(d)} = \textsf{average}(\bU_\post^{(d)} \cdot v_i )$.

\bigskip
\noindent {\em Observations.}
Let  
(i) $\Yipre = \bU^{(0)}_\pre \cdot v_i + \varepsilon_{\pre, i}$; 
(ii) $\bYpred =\bU^{(0)}_\pre \cdot \bV^\top_{\cI^{(d)}} + \bE_{\pre, \Ic^{(d)}}$; 
and
(iii) $\bYpostd = \bU^{(d)}_\post \cdot \bV^\top_{\cI^{(d)}} + \bE_{\post, \Ic^{(d)}}$,
where the entries of $(\varepsilon_{\pre, i}, \bE_{\pre, \Ic^{(d)}}, \bE_{\post, \Ic^{(d)}} )$ are i.i.d. samples from a standard normal. 

\paragraph{Simulation results.} 
From  $(\Yipre, \bYpred)$, we learn $\hw^{(i,d,\Omega)}$ via the modified \SI-\PCR~estimator, where $\Omega$ is chosen to index the first $r$ donors of $\cI^{(d)}$; hence, $| \Omega | = r$. 
We define our estimate as $\htheta_i^{(d, \Omega)} = \textsf{average}( \bYpostd^{\Omega} \cdot \hw^{(i,d,\Omega)} )$ and construct its corresponding confidence interval as per \eqref{eq:finite_sample_CI}. 
Table~\ref{table:coverage} reports on the empirical coverage probabilities and average interval lengths for $\theta_i^{(d)}$ at each $T_0$ over $50$ trials, where each trial consists of an independent draw of the latent factors and the observations of size $100$; this amounts to $5000$ total simulation repeats per dimension. 
As suggested by Lemma~\ref{lemma:cond.iii.compressed}, we observe that our confidence interval consistently achieves close to the nominal targets of $90\%$ and $95\%$ across all values of $T_0$, albeit the interval lengths can be large due to the poorer prediction quality compared to the standard \SI-\PCR~estimator of \eqref{eq:pcr}. 

\begin{table} [!t]
\centering 
\caption{Coverage results for the modified \SI-\PCR~estimator across 5000 replications.} 
\begin{tabular}{@{}lcccc@{}}
\toprule
\midrule 
\multirow{2}{*}{Size of $\Tc_\pre$}
& \multicolumn{2}{c}{$90\%$ nominal coverage}   
& \multicolumn{2}{c}{$95\%$ nominal coverage}   
\\ 
\cmidrule(l){2-3} 
\cmidrule(l){4-5} 
& Coverage probability  & Interval length 
& Coverage probability  & Interval length 
\\
\midrule 
$T_0 = 200$   	 
& 0.88	& 15.66
& 0.94	& 18.66	 
\\
\midrule
$T_0 = 400$      	 
& 0.88	& 18.90
& 0.94	& 22.52
\\
\midrule
$T_0 = 600$      	 
& 0.88	& 8.90
& 0.94	& 10.61
\\
\midrule
$T_0 = 800$    	 
& 0.89	& 7.25
& 0.94	& 8.64 
\\
\midrule
$T_0 = 1000$     	 
& 0.90	& 15.58
& 0.95	& 18.57 
\\
\bottomrule
\end{tabular}
\label{table:coverage}
\end{table}

\subsection{Proof of Theorem \ref{thm:normality}} 
\begin{proof} 
For ease of notation, we suppress the conditioning on $\Ec$ for the remainder of the proof. 
We scale the left-hand side (LHS) of \eqref{eq:inf.1} by $\sqrt{T_1}/ (\sigma \|\tw^{(i,d)}\|_2)$ and 
analyze each of the three terms on the right-hand side (RHS) of \eqref{eq:inf.1} separately.  
To address the first term on the RHS of \eqref{eq:inf.1}, observe that \eqref{eq:cond.iii} immediately gives
\begin{align} \label{eq:normality.1} 
    \frac{1}{\sqrt{T_1}\sigma \|\tw^{(i,d)}\|_2} \sum_{t > T_0}  \langle \Ex[Y_{t, \Ic^{(d)}}], \Pc_{V_\pre} \Delta^{(i,d)} \rangle
    = o_p(1). 
\end{align}
For the second term on the RHS of \eqref{eq:inf.1}, note that $\langle \varepsilon_{t, \Ic^{(d)}},  \tw^{(i,d)} \rangle$ are independent across $t$.
Note that $ \langle \varepsilon_{t, \Ic^{(d)}}, \tw^{(i,d)} / (\sigma \|\tw^{(i,d)}\|_2)  \rangle$ is a sub-Gaussian random variable with variance equal to $1$ and sub-Gaussian norm less than some constant $C > 0$.
It is easy to verify that the Lyapunov condition \cite{Billingsley} for $\delta = 1$ holds, and hence the Lyapunov Central Limit Theorem yields
\begin{align} \label{eq:normality.2} 
    \frac{1}{\sqrt{T_1}} \sum_{t > T_0} \Big\langle \varepsilon_{t, \Ic^{(d)}}, \frac{\tw^{(i,d)}}{\sigma \|\tw^{(i,d)}\|_2} \Big\rangle
    \xrightarrow{d} 
    \mathcal{N}\big( 0,  1\big).
\end{align}
as $T_1 \to \infty$.

For the third term on the RHS of \eqref{eq:inf.1}, we scale the RHS of \eqref{eq:consistency.3} by $\sqrt{T_1}/ (\sigma \|\tw^{(i,d)}\|_2)$ and recall \eqref{eq:cond.ii}. 
This yields 
\begin{align} \label{eq:normality.3} 
    \frac{1}{\sqrt{T_1} \sigma \|\tw^{(i,d)}\|_2 } \sum_{t > T_0} \langle \varepsilon_{t, \Ic^{(d)}}, \Delta^{(i,d)}  \rangle
    &= o_p(1).
\end{align}
Collecting \eqref{eq:normality.1}, \eqref{eq:normality.2}, \eqref{eq:normality.3}, we establish \eqref{eq:gaussian.limit}.

\end{proof}

\subsection{Proof of Lemma~\ref{lemma:cond.iii.compressed}} 
We state the proof of Lemma~\ref{lemma:cond.iii.compressed}, which is an adaptation of the proof of Lemma~\ref{lemma:pcr_aos2} in Appendix~\ref{sec:consistency.1}.
 
\medskip 
\begin{proof} 
Below we establish \eqref{eq:cond.ii_PCR_variant} and \eqref{eq:cond.iii_PCR_variant} hold.

\begin{itemize}

\item[(a)] {\em Proof of \eqref{eq:cond.ii_PCR_variant}.} 
Follows from conditions (v) and (vi) in the lemma statement. 

\item[(b)] {\em Proof of \eqref{eq:cond.iii_PCR_variant}.} 
By condition (ii), the rowspaces of both $\bYpred^{(r_\pre, \Omega)}$ and $\bMpred^\Omega$ span all of $\Rb^{r_\pre}$.
Hence, the orthogonal projection matrices onto their rowspaces obey 
\begin{align}
	\left( \bYpred^{(r_\pre, \Omega)}\right)^\dagger   \bYpred^{(r_\pre, \Omega)} = \left(\bMpred^\Omega \right)^\dagger \bMpred^\Omega = \bI_{r_\pre}, \label{eq:projection.1} 
\end{align}
where $\bI_{r_\pre}$ denotes the $r_\pre \times r_\pre$-dimensional identity matrix.  
%
%
The remainder of this proof is dedicated to bounding $\Delta^{(i,d,\Omega)} = \| \hw^{(i,d,\Omega)} - w^{(i,d,\Omega)} \|_2$ using the proof of Lemma~\ref{lemma:pcr_aos2} in Appendix~\ref{sec:consistency.1} as a guide. 

Using \eqref{eq:projection.1} and the arguments that led to \eqref{eq:thm1.2}, we obtain 
\begin{align} 
 \| \Delta^{(i,d,\Omega)}  \|_2^2 
& \leq \frac{4}{(\hs^\Omega_{r_\pre})^2} \left( \| \bYpred^{(r_\pre, \Omega)} - \bMpred^\Omega \|_{2,\infty}^2 \cdot \|\tw^{(i,d,\Omega)}\|_1^2
+\langle \bYpred^{(r_\pre, \Omega)} \cdot \Delta^{(i,d,\Omega)}, \varepsilon_{\pre, n} \rangle\right). \label{eq:compressed.1} 
\end{align}
where $\hs^\Omega_{k}$ is defined as the $k$-th singular value of $\bYpred^{(r_\pre, \Omega)}$ while $\varepsilon_{\pre,n}$ maintains the same definition as set in Appendix~\ref{sec:proof_consistency}.
By construction, it follows that 
\begin{align}
	\| \bYpred^{(r_\pre, \Omega)} - \bMpred^\Omega \|_{2,\infty}^2 \le \| \bYpred^{r_\pre} - \bMpred \|_{2,\infty}^2 \label{eq:compressed.2} 
\end{align} 
since $\bYpred^{(r_\pre, \Omega)}$ and $\bMpred^\Omega$ are sub-matrices of $\bYpred^{r_\pre}$ and $\bMpred$, respectively. 

Let $s^\Omega_{k}$ denote the $k$-th singular value of $\bMpred^\Omega$. 
Further, without loss of generality, suppose $\Omega$ corresponds to the first $r_\pre$ entries of $\cI^{(d)}$. 
This allows us to write 
\begin{align}
	\bYpred^{(r_\pre, \Omega)} = \bYpred^{r_\pre} \bA,
	\qquad
	\bMpred^\Omega = \bMpred \bA, \label{eq:transformation}
\end{align} 
where $\bA = [\bI_{r_\pre}, 0_{r_\pre \times (\numdonors-r_\pre)}]^\top$. 
Combining \eqref{eq:transformation} with Lemma~\ref{lemma:weyl}, we have 
\begin{align} 
	| \hs^\Omega_{r_\pre} - s^\Omega_{r_\pre} | &\le \| \bYpred^{(r_\pre, \Omega)} - \bMpred^\Omega \|_\txtop 
	\\ &= \| (\bYpred^{r_\pre} - \bMpred ) \cdot \bA \|_\txtop
	\\ &\le \| \bYpred^{r_\pre} - \bMpred  \|_\txtop  &&\because \| \bA \|_\txtop = 1
	\\ &\le \| \bYpred^{r_\pre} - \bYpred \|_\txtop + \| \bYpred - \bMpred \|_\txtop 
	\\ &= \hs_{r_\pre + 1} + \| \bYpred - \bMpred \|_\txtop, \label{eq:s.compressed.1} 
\end{align} 
where we recall $\hs_k$ is the $k$-th singular value of $\bYpred$. 
Let us define $s_k$ analogously to $\hs_k$ with respect to $\bMpred$.  
By condition (i), $s_{r_\pre+1} = 0$ and thus 
\begin{align}
	\hs_{r_\pre+1} &= \hs_{r_\pre+1} - s_{r_\pre+1} 
	\\ &\le \| \bYpred - \bMpred \|_\txtop. &&\because \text{Lemma~\ref{lemma:weyl}}  \label{eq:s.compressed.2} 
\end{align} 
Applying Assumption \ref{assumption:noise} and Lemma \ref{lemma:subg_matrix} to \eqref{eq:s.compressed.1} and \eqref{eq:s.compressed.2} yields w.p. at least $1-2\exp(-\zeta^2)$, 
\begin{align}
	| \hs^\Omega_{r_\pre} - s^\Omega_{r_\pre} | &\le 2 \cdot \| \bYpred - \bMpred \|_\txtop
	\\ &\le C \sigma \left( \sqrt{T_0} + \sqrt{\numdonors} + \zeta \right) \label{eq:s.compressed.3}
\end{align}  
for any $\zeta > 0$. 
Next, we follow the arguments that led to \eqref{eq:dd.1}.
In particular, we plug \eqref{eq:compressed.2} and \eqref{eq:s.compressed.3} into \eqref{eq:compressed.1}, and invoke Lemmas~\ref{lemma:coeff_bound}, \ref{lemma:hsvt}, \ref{lemma:annoying}, and conditions (iii)--(iv) to conclude
\begin{align} \label{eq:new.pcr.rate}
	  \Delta^{(i,d,\Omega)}  &= O_p \left( 
	  \frac{r_\pre \sqrt{ \log(T_0 N_d)}  }{ \sqrt{N_d}} 
	  + \frac{1}{T_0^{1/4}} \right). 
\end{align} 
Then,
\begin{align}
&\frac{1}{\sqrt{T_1} \|w^{(i,d,\Omega)}\|_2} \sum_{t > T_0} \left \langle \Ex[Y^\Omega_{t, \Ic^{(d)}}],  \left( \bM^\Omega_{\pre, \Ic^{(d)}} \right)^\dagger  \bM^\Omega_{\pre, \Ic^{(d)}}\left(\hw^{(i,d,\Omega)} - w^{(i,d,\Omega)} \right) \right \rangle,
\\ &\quad \le \frac{1}{\sqrt{T_1} \|w^{(i,d,\Omega)}\|_2} \sum_{t > T_0} \| \Ex[Y^\Omega_{t, \Ic^{(d)}}] \|_2 \cdot \| \Delta^{(i,d,\Omega)} \|_2,
\\ &\quad\le \frac{1}{\sqrt{T_1} \|w^{(i,d,\Omega)}\|_2} \sum_{t > T_0} \sqrt{r_\pre} \cdot \| \Delta^{(i,d,\Omega)} \|_2,
%
%
\\ &\quad= O_p \left( \frac{ \sqrt{T_1} r^{2}_\pre \sqrt{ \log(T_0 N_d)}  }{ \sqrt{N_d}} + \frac{ \sqrt{T_1} r_\pre}{T_0^{1/4}} \right), 
\\ &\quad= o_p(1),
\end{align}
where the second to last line follows from \eqref{eq:new.pcr.rate} and condition (v), and the last line follows from condition (vii).
\end{itemize}
\end{proof}